\documentclass[EJP, preprint]{ejpecp} 

\usepackage{enumerate}  
\usepackage{bm}

\SHORTTITLE{IDS asymptotics for Gibbsian singular potentials}

\TITLE{Asymptotics of the IDS for Schr\"{o}dinger operators with singular potentials and Gibbs point processes
   } 

\AUTHORS{%
Yuta Nakagawa\footnote{Graduate School of Human and Environmental Studies, Kyoto University, Japan.
    \BEMAIL{nakagawa.yuta.58n@st.kyoto-u.ac.jp}
  }}


\AMSSUBJ{Primary 82B44, Secondary 60G55; 60G60; 60K35} 

\SUBMITTED{} 

\VOLUME{0}
\YEAR{2026}
\PAPERNUM{0}
\DOI{10.1214/YY-TN}

\ABSTRACT{
The asymptotic behavior of the integrated density of states (IDS), \(N(E)\), is investigated for random Schrödinger operators with a single-site potential \(V\) satisfying \(\essinf V = -\infty\). Under the assumption that the underlying point process is a Gibbs point process with repulsive pairwise interactions, the leading term of \(\log N(E)\) as \(E \to -\infty\) is determined using a periodic approximation method. It is shown that repulsive pairwise interactions lead to a significantly faster decay of \(N(E)\) compared to the Poisson case. Furthermore, configurations with multiple clusters can provide the dominant contribution to the IDS in the Gibbs setting, contrasting with the single-cluster dominance typically observed in Poisson models. Finally, refined estimates of the leading constants are provided for specific classes of potentials, including those with multiple singularities.
}

\newcommand{\veps}{\varepsilon} 
\newcommand{\borel}{\mathcal{B}(\mathbb{R}^d)}

\DeclareMathOperator{\tr}{tr}
\DeclareMathOperator{\diam}{diam}
\DeclareMathOperator{\essinf}{ess\,inf}
\DeclareMathOperator{\esssup}{ess\,sup}
\DeclareMathOperator{\Int}{Int}
\DeclareMathOperator{\supp}{supp}

\let\Im\relax
\DeclareMathOperator{\Im}{Im}

\begin{document}



\section{Introduction}\label{sec:intro}
We consider the random Schr\"{o}dinger operator on \(L^2(\mathbb{R}^d,dx)\) defined by
\begin{equation}\label{eq:operator}
H_{\omega}= -\Delta + V_{\omega},
\end{equation}
where
\begin{equation} 
V_{\omega}=\sum_{x\in\omega}V(\,\cdot\,-x).
\end{equation}
Here, \(V\), called the \emph{single-site potential}, is a real-valued measurable function on \(\mathbb{R}^d\), and \(\omega\subset\mathbb{R}^d\) is a configuration of points representing a realization of a point process \(\Gamma\) (see Section \ref{subsec:def}). The \emph{integrated density of states (IDS)} \(N(E)\) of the operator \(H_{\omega}\) is the function on \(\mathbb{R}\) formally given by
\begin{equation}
\lim_{L\to\infty}\frac{1}{L^d}\# \{\text{eigenvalues of } H_{\omega,L} \text{ less than or equal to } E\},
\end{equation}
where \(H_{\omega,L}\) is the operator \(H_{\omega}\) restricted to the box \((-L/2,L/2)^d \subset \mathbb{R}^d\) with Dirichlet boundary conditions. When \(\Gamma\) is stationary and ergodic (see Section \ref{subsec:properties}), it is known that both \(N(E)\) and the spectrum of \(H_{\omega}\) are almost surely independent of the realization \(\omega\). The IDS \(N(E)\) is non-decreasing and increases only on the spectrum. We refer to \cite{CL, PF} for the precise definition and further properties of the IDS.

The asymptotic behavior of \(\log N(E)\) as \(E\to-\infty\) is well known for the case where the single-site potential \(V\not\equiv 0\) is a nonpositive, bounded, and continuous function with compact support.
If the point process \(\Gamma\) is a Poisson point process (see Section \ref{subsec:def} for the definition), it is known that 
\begin{equation}\label{eq:classical}
\log N(E)\sim -\frac{1}{|\min V|}|E|\log|E|\quad(E\to-\infty),
\end{equation} 
where \(f(E)\sim g(E)\ (E\to-\infty)\) means that \(f(E)/g(E)\) converges to one as \(E\to-\infty\). We refer to \cite{Pastur} for this classical result.

When \(\Gamma\) is a stationary and ergodic Gibbs point process, which is characterized by interactions between points (see Section \ref{subsec:def} for the precise definition), the behavior of the IDS depends on the type of interaction. If the interaction is sufficiently weak, the corresponding IDS \(N(E)\) exhibits the same behavior as in \eqref{eq:classical} (see \cite{Nakagawa}).
In contrast, for a Gibbs point process with a pairwise energy function (see Example \ref{ex:pairwise}), the corresponding IDS can decay significantly faster. For instance, if the energy of the points \(\{x_i\}\) is given by 
\begin{equation}\label{eq:intro pairwise}
a\sum_{i<j}\mathbb{1}_{[0,R]}(|x_i-x_j|)
\end{equation}
for some \(a, R>0\), it holds that
\begin{equation}
\log N(E)\sim -\frac{a}{2\|V\|_R^2} |E|^2\quad (E\to-\infty),
\end{equation}
where \(\mathbb{1}_A\) denotes the indicator function of a set \(A\), and the constant \(\|V\|_R^2\) is determined by the single-site potential \(V\) and the interaction range \(R\).
See \cite{Nakagawa} for further details.

Various models concerning the asymptotic behavior of the IDS as \(E\to-\infty\) have been investigated. The one-dimensional case with a negative delta potential \(V=-c\delta\ (c>0)\) was discussed in \cite{KP, LGP}. The case of Poisson point interactions on \(\mathbb{R}^3\) was studied in \cite{KMN}. Beyond the Poisson setting, the case of a randomly perturbed lattice, where the single-site potential is bounded below, was investigated in \cite{FU}. 

In this paper, we investigate the asymptotic behavior of \(\log N(E)\) as \(E\to-\infty\) for Gibbs point processes, where the single-site potential \(V\) satisfies \(\essinf V=-\infty\) along with certain integrability and exponential decay conditions. An example of such a potential in three dimensions \((d=3)\) is the screened Coulomb potential, given by \(V(x) = -e^{-|x|} / |x|\).  For the case of a Poisson point process under the same conditions on \(V\), it was proved in \cite{KP} that the IDS satisfies
\begin{equation}\label{eq:Poisson IDS}
\log N(E)\asymp - g(|E|)\log g(|E|)\quad (E\to-\infty),
\end{equation}
where \(g\) is an increasing function determined by \(V\) (see Section \ref{sec:weak} for the precise definition of \(g\)). Here, \(A(E)\asymp B(E)\ (E\to-\infty)\) means that the ratio \(A(E)/B(E)\) is bounded from above and below by two positive constants as \(E\to-\infty\).

As discussed in \cite{KP}, for singular potentials (the quantum case), the precision of the standard Dirichlet-Neumann bracketing is not sufficient. To obtain more precise bounds, one must employ the periodic approximation method. Our approach is based on this method. By extending it to the case of Gibbs point processes and single-site potentials satisfying \(\essinf V=-\infty\), we prove the following results.
First, we show that for a Gibbs point process with sufficiently weak interactions, the asymptotic behavior of the IDS is analogous to that for the Poisson point process studied in \cite{KP}; specifically, the same asymptotic relation as \eqref{eq:Poisson IDS} holds. This result is presented in Section \ref{sec:weak}.
In contrast, for the Gibbs point process with the pairwise energy function given by \eqref{eq:intro pairwise}, we show that
\begin{equation}
\log N(E) \asymp -g(|E|)^2\quad (E\to-\infty).
\end{equation}
This result, established in Theorem \ref{thm:pairwise}, demonstrates that the IDS decays significantly faster than in the Poisson case due to the repulsive effect of the interactions between points.
In addition, if the support of \(V\) is contained in an open ball of diameter \(R\), it holds that
\begin{equation}
\log N(E) \sim -\frac{a}{2}\,g(|E|)^2 \quad (E\to-\infty).
\end{equation}
Refer to Example \ref{ex:Strauss} for the details of this result.
We remark that, in the Poisson case, the leading term of \(\log N(E)\) for such a single-site potential \(V\) has not been determined yet, except for certain classes of \(V\) discussed in \cite{KP}.

Furthermore, we investigate the case where the single-site potential \(V\) is represented as a linear combination of translates of a function with small compact support. In this setting, under certain conditions on the range \(R\) and the single-site potential \(V\), we determine the leading term of \(\log N(E)\). The resulting asymptotic behavior reflects how deep potential wells are formed: whereas the leading term in the Poisson case is dominated by the event where a single cluster of points forms a deep well, repulsive interactions make the probability of such a configuration highly suppressed. Instead, the dominant contribution to the IDS arises from the event that the points form several distinct clusters. A similar feature regarding the number of clusters was also reported in the author's previous work \cite{Nakagawa} in the case of bounded potentials. We show that the leading term of \(\log N(E)\) as \(E\to-\infty\) for the class of potentials considered here is characterized by these multi-cluster configurations, and the constant in the leading term is determined in Section \ref{sec:improve}.

This paper is organized as follows. In Section \ref{sec:GPP}, we introduce Gibbs point processes and their properties (see, e.g., \cite{Dereudre, Preston, Ruelle}). In Section \ref{sec:PA}, we modify the periodic approximation method established in \cite{KP} and apply it to the case where the point process \(\Gamma\) is a Gibbs point process.
In Section \ref{sec:weak}, we show that the asymptotic behavior of the IDS for a Gibbs point process with weak interactions is analogous to that for a Poisson point process. In Section \ref{sec:pairwise}, we determine the leading term of the asymptotic behavior of the IDS for a Gibbs point process with a pairwise energy function satisfying suitable conditions. Finally, in Section \ref{sec:improve}, we provide a refined estimate for a specific class of single-site potentials, including those with multiple singularities.

\section{Gibbs point process}\label{sec:GPP}
In this section, we provide an overview of Gibbs point processes, including their definitions and basic properties. 
After introducing the formal framework in Section \ref{subsec:def}, we recall relevant properties, such as disagreement couplings, in Section \ref{subsec:properties}.

\subsection{Definition and notation}\label{subsec:def}

Let \((\mathcal{C},\mathcal{F})\) denote the space of all subsets of \(\mathbb{R}^d\) without accumulation points, equipped with the \(\sigma\)-algebra \(\mathcal{F}\) generated by counting functions \(\{m(\,\cdot\, ,\Lambda)\}_{\Lambda\in\borel}\), where \(m(\,\cdot\, ,\Lambda)\) is the function on \(\mathcal{C}\) defined by \(\mathcal{C}\ni\omega\mapsto\#(\omega\cap\Lambda)\).
For every \(\Lambda\in\borel\), we set \(\mathcal{C}_{\Lambda}=\{\omega_{\Lambda}\mid \omega\in\mathcal{C}\}\), where \(\omega_\Lambda=\omega\cap\Lambda\).
The space \(\mathcal{C}_{\Lambda}\) is equipped with the \(\sigma\)-algebra \(\mathcal{F}_{\Lambda}\) generated by \(\{m(\,\cdot\, ,\Lambda')\}_{\Lambda'\in\borel,\ \Lambda'\subset\Lambda}\).

A \emph{point process} on \(\Lambda\in\borel\) is a \(\mathcal{C}_{\Lambda}\)-valued random variable.
For every \(\mu>0\) and \(\Lambda\in\borel\), let \(P^{Poi(\mu)}_{\Lambda}\) denote the distribution of the \emph{Poisson point process} on \(\Lambda\) with intensity \(\mu\). That is, \(P^{Poi(\mu)}_{\Lambda}\) is the probability measure on \(\mathcal{C}_{\Lambda}\) characterized by
\begin{itemize}
\item for any \(n\in\mathbb{Z}_{>0}\) and disjoint \(\Lambda'_1,\ldots,\Lambda'_n \subset \Lambda\), the random variables \(m(\,\cdot\, ,\Lambda'_1),\ldots,m(\,\cdot\, ,\Lambda'_n)\) are mutually independent;
\item \(P^{Poi(\mu)}_{\Lambda}(m(\,\cdot\, ,\Lambda')=n)=e^{-\mu|\Lambda'|} (\mu|\Lambda'|)^n / n!\) for all bounded \(\Lambda' \subset \Lambda\) and \(n\in\mathbb{Z}_{\geq 0}\),
\end{itemize}
where \(|\Lambda'|\) denotes the Lebesgue measure of \(\Lambda'\). For simplicity, we write \(P^{Poi}_{\Lambda}=P^{Poi(1)}_{\Lambda}\), \(P^{Poi(\mu)}=P^{Poi(\mu)}_{\mathbb{R}^d}\) and \(P^{Poi}=P^{Poi(1)}\).

An \emph{energy function} \(U\) is a measurable function from \(\mathcal{C}_f\) to \(\mathbb{R}\cup\{+\infty\}\), where \(\mathcal{C}_f\) is the subspace of \(\mathcal{C}\) defined by \(\mathcal{C}_f = \{\omega\in\mathcal{C} \mid \#\omega<\infty\}\). We assume that \(U\) can be represented as
\begin{equation}
U(\omega)=\sum_{\eta\subset\omega}\Phi(\eta),
\end{equation}
for some measurable function \(\Phi\) on \(\mathcal{C}_f\). This function \(\Phi\) is called the \emph{interaction function}. Note that this assumption is satisfied if \(U(\omega\cup\{x\})=+\infty\) for any \(x\in\mathbb{R}^d\) and for any \(\omega\in\mathcal{C}_f\) such that \(U(\omega)=+\infty\), which follows from the M\"{o}bius inversion formula (see, e.g., \cite{Preston}).

In addition, we assume that the interaction function \(\Phi\) satisfies the following conditions:
\begin{description}
\item[\phantomsection\label{cd:U1}\textbf{(U1)}] 
\begin{enumerate}[(i)]
\item \(\Phi(\emptyset)=0\);
\item \(\Phi(\tau_x \omega)=\Phi(\omega)\) for any \(x\in\mathbb{R}^d\) and \(\omega\in\mathcal{C}_f\), where \(\tau_{x}\) is the translation by \(x\);
\item for some \(R>0\), \(\Phi(\omega)=0\) whenever \(\diam \omega>R\).
\end{enumerate}
\end{description}
Here, \(\diam \omega\) denotes the diameter of \(\omega\) in the Euclidean space \(\mathbb{R}^d\).

By Condition \hyperref[cd:U1]{\textbf{(U1)}}(iii), for \(\gamma\in\mathcal{C}\), \(\omega\in\mathcal{C}_{\Lambda}\), and a bounded \(\Lambda\in\borel\), we can define
\begin{equation}
U_{\Lambda,\gamma}(\omega) = \sum_{\substack{\eta\subset\omega\cup\gamma_{\Lambda^c}\\[0.3ex] \eta\cap\omega\neq\emptyset}}\Phi(\eta).
\end{equation}
We remark that \(U_{\Lambda,\gamma}(\omega)=U_{\Lambda,\gamma_{\Lambda+B(0,R)}}(\omega)\), where \(B(x,R)\) denotes the closed ball centered at \(x\in\mathbb{R}^d\) with radius \(R\), and \(\Lambda_1+\Lambda_2\) denotes the Minkowski sum of \(\Lambda_1\) and \(\Lambda_2\). 

For every \(x\in\mathbb{R}^d\) and \(\gamma\in\mathcal{C}\), we define the \emph{local energy function} by
\begin{equation}
u(x;\gamma)= U_{\{x\},\gamma}(\{x\}),
\end{equation}
which represents the change in energy when adding a point \(x\) to the configuration \(\gamma\). 

Furthermore, we assume that the energy function satisfies the following stability condition:
\begin{description}
\item[\phantomsection \textbf{(U2)} \label{cd:U2}] there exists a constant \(b\in\mathbb{R}\) such that \(u(x;\gamma)\geq b\) for any \(x\in\mathbb{R}^d\) and any \(\gamma\in\mathcal{C}\).
\end{description}

Under Condition \hyperref[cd:U2]{\textbf{(U2)}}, for any bounded \(\Lambda\in\borel\), \(\gamma\in\mathcal{C}\), and distinct points \(x_1,\ldots, x_k\in \Lambda\), we obtain
\begin{equation}\label{eq:stable}
U_{\Lambda,\gamma}(\{x_1, \ldots , x_k\})=\sum_{j=1}^k u(x_j;\gamma_{\Lambda^c}\cup\{x_1,\ldots,x_{j-1}\})
\geq kb.
\end{equation}

Formally, a Gibbs point process is defined as a point process whose distribution satisfies the Dobrushin-Lanford-Ruelle (DLR) equation (see, e.g., \cite{Dereudre, DV}).
\begin{definition}[Gibbs point process]
We say that a point process is a \emph{Gibbs point process} for an energy function \(U\) if for any bounded \(\Lambda\in\borel\) with positive Lebesgue measure and any bounded measurable function \(f\) on \(\mathcal{C}\), its distribution \(P^{Gib}\) satisfies the \emph{DLR equation}:
\begin{equation}\label{eq:DLR}
\int_{\mathcal{C}}f(\eta)P^{Gib}(d\eta)=\int_{\mathcal{C}}\int_{\mathcal{C}_{\Lambda}}f(\eta\cup\gamma_{\Lambda^c})P^{Gib}_{\Lambda,\gamma}(d\eta)P^{Gib}(d\gamma),
\end{equation}
where \(P^{Gib}_{\Lambda,\gamma}\) is the probability measure on \(\mathcal{C}_{\Lambda}\) defined by
\begin{equation}
dP^{Gib}_{\Lambda,\gamma} = \frac{1}{Z_{\Lambda,\gamma}}e^{-U_{\Lambda,\gamma}}dP^{Poi}_{\Lambda},
\end{equation}
and \(Z_{\Lambda,\gamma}\) is the normalization constant defined by
\begin{equation}\label{eq:partition function}
Z_{\Lambda,\gamma}=\int_{\mathcal{C}_{\Lambda}}e^{-U_{\Lambda,\gamma}(\eta)}P^{Poi}_{\Lambda}(d\eta).
\end{equation}
\end{definition}

Note that under Conditions \hyperref[cd:U1]{\textbf{(U1)}} and \hyperref[cd:U2]{\textbf{(U2)}}, we have \(0 < Z_{\Lambda,\gamma}<+\infty\), and a corresponding Gibbs point process exists (see, e.g., \cite{Dereudre, DV}).

To extend the periodic approximation method to the Gibbs setting, we introduce Gibbs point processes with periodic boundary conditions. This construction follows the approach for Gibbs distributions on the lattice \(\mathbb{Z}^d\) with periodic boundary conditions as detailed in \cite{Georgii}.

For \(x \in \mathbb{R}^d\) and \(n > 0\), let \(\Lambda_n(x)\) denote the \(d\)-dimensional cube of side length \(n\) centered at \(x\), i.e., \(\Lambda_n(x) = \prod_{i=1}^{d}(x_i-n/2, x_i+n/2]\), where we put \(x=(x_1, \ldots, x_d)\). For simplicity, we write \(\Lambda_n = \Lambda_n(0)\).
We consider an energy function \(U\) satisfying Conditions \hyperref[cd:U1]{\textbf{(U1)}} and \hyperref[cd:U2]{\textbf{(U2)}}. Fix \(n> 2R\). Let \(\pi_n\) be the map from \(\mathcal{C}_{\Lambda_n}\) to \(\mathcal{C}\) defined by
\begin{equation}
\pi_n(\omega)=\bigcup_{x\in n\mathbb{Z}^d}\tau_x \omega.
\end{equation}
We define \(\tau^{per}_{n,x}=\pi^{-1}_n \circ \tau_x \circ \pi_n\) for \(x\in\mathbb{R}^d\). This map can be identified with the translation on the \(d\)-dimensional torus \(\mathbb{R}^d / n\mathbb{Z}^d\).
Since \(\tau_x\circ\pi_n(\omega)\in\pi_n(\mathcal{C}_{\Lambda_n})\), this map is well-defined. 
We define the \emph{periodic interaction function} \(\Phi^{per}_n\) on \(\mathcal{C}_{\Lambda_n}\) by
\begin{equation}
\Phi^{per}_n(\omega) =
\begin{cases}
\Phi(\tau_{n,x}^{per} \omega) & \text{if \(\diam(\tau^{per}_{n,x} \omega) \leq R\) for some \(x\in\mathbb{R}^d\)}, \\
0 & \text{otherwise.}
\end{cases}
\end{equation}
This is well-defined due to the condition \(n>2R\) and the translation invariance \hyperref[cd:U1]{\textbf{(U1)}}(ii).
The \emph{periodic energy function} \(U^{per}_{\Lambda_n}\) on \(\mathcal{C}_{\Lambda_n}\) is defined by
\begin{equation}
U^{per}_{\Lambda_n}(\omega)= \sum_{\eta\subset \omega}\Phi^{per}_n(\eta).
\end{equation}
We remark that the following translation invariance holds: for all \(x\in\mathbb{R}^d\) and \(\omega\in\mathcal{C}_{\Lambda_n}\),
\begin{equation}\label{eq:per stationary}
U^{per}_{\Lambda_n}(\tau^{per}_{n,x}\omega)= U^{per}_{\Lambda_n}(\omega).
\end{equation}

Furthermore, for any measurable subset \(\Lambda\subset\Lambda_n\), \(\omega\in\mathcal{C}_{\Lambda}\), and \(\gamma\in\mathcal{C}_{\Lambda_n}\),
we set
\begin{equation}
U^{per}_{\Lambda_n,\Lambda,\gamma}(\omega)=\sum_{\substack{\eta\subset\omega\cup\gamma_{\Lambda^c}\\[0.3ex] \eta\cap\omega\neq\emptyset}}\Phi^{per}_n(\eta).
\end{equation}
Note that for any \(\gamma\in\mathcal{C}_{\Lambda_n}\) and any \(\Lambda\subset\Lambda_{n-2R}\), the following identity holds on \(\mathcal{C}_{\Lambda}\):
\begin{equation}\label{eq:per local}
U^{per}_{\Lambda_n,\Lambda,\gamma}=U_{\Lambda,\gamma}.
\end{equation} 

The following lemma shows that the periodic energy function \(U^{per}_{\Lambda_n}\) inherits stability from \(U\).
\begin{lemma}\label{lem:per stable}
Assume that an energy function \(U\) satisfies Conditions \hyperref[cd:U1]{\textbf{(U1)}} and \hyperref[cd:U2]{\textbf{(U2)}}. For any \(n>2R\) and \(\omega\in\mathcal{C}_{\Lambda_n}\), we have 
\begin{equation}\label{eq:per stable}
U^{per}_{\Lambda_n}(\omega)\geq b\# \omega.
\end{equation}
\end{lemma}

\begin{proof}
If \(\omega=\emptyset\), we obtain \(U^{per}_{\Lambda_n}(\omega)=0\).
For \(\omega\neq\emptyset\), write \(\omega=\omega'\cup\{x\}\), where \(x\notin\omega'\).
From \eqref{eq:per stationary} and \eqref{eq:per local}, we obtain
\begin{equation}
U^{per}_{\Lambda_n}(\omega)
=U^{per}_{\Lambda_n}(\tau^{per}_{n,-x}\omega' \cup\{0\})
=U^{per}_{\Lambda_n}(\omega')+u(0;\tau^{per}_{n,-x}\omega').
\end{equation}
Applying Condition \hyperref[cd:U2]{\textbf{(U2)}}, we have \(U^{per}_{\Lambda_n}(\omega)\geq U^{per}_{\Lambda_n}(\omega') + b\).
By repeating this procedure, we obtain \eqref{eq:per stable}.
\end{proof}

For the energy function \(U\), we define the distribution of the corresponding Gibbs point process with periodic boundary conditions on \(\mathcal{C}_{\Lambda_n}\) by
\begin{equation}
dP^{per}_{\Lambda_n}=\frac{1}{Z^{per}_{\Lambda_n}}e^{-U^{per}_{\Lambda_n}}dP^{Poi}_{\Lambda_n},
\end{equation}
where \(Z^{per}_{\Lambda_n}\) is the normalization constant given by
\begin{equation}
Z^{per}_{\Lambda_n}=\int_{\mathcal{C}_{\Lambda_n}}e^{-U^{per}_{\Lambda_n}}dP^{Poi}_{\Lambda_n}.
\end{equation}
Note that Lemma \ref{lem:per stable} implies that \(\exp(-|\Lambda_n|)\leq Z^{per}_{\Lambda_n}\leq \exp(|\Lambda_n|(e^{-b}-1))\).
Combining \eqref{eq:per local} with the DLR equation for finite volumes (see \cite{Dereudre}), we see that for any \(\Lambda\subset\Lambda_{n-2R}\), the probability measure \(P^{per}_{\Lambda_n}\) satisfies the following DLR equation: 
\begin{equation}\label{eq:per DLR}
\int_{\mathcal{C}_{\Lambda_n}}f dP^{per}_{\Lambda_n}= \int_{\mathcal{C}_{\Lambda_n}}\int_{\mathcal{C}_{\Lambda}} f(\eta\cup\gamma_{\Lambda^c})P^{Gib}_{\Lambda,\gamma}(d\eta) P^{per}_{\Lambda_n}(d\gamma),
\end{equation}
where \(f\) is a bounded measurable function on \(\mathcal{C}_{\Lambda_n}\). 

Furthermore, since \(P^{Poi}_{\Lambda_n}\) is invariant under \(\tau^{per}_{n,x}\), it follows from \eqref{eq:per stationary} that 
\begin{equation}\label{eq:per GPP stationary}
P^{per}_{\Lambda_n}(\tau^{per}_{n,x}A)=P^{per}_{\Lambda_n}(A)
\end{equation}
for all \(x\in\mathbb{R}^d\) and \(A\in\mathcal{F}_{\Lambda_n}\).

\subsection{Properties of the Gibbs point process}\label{subsec:properties}

In this subsection, we summarize the fundamental properties of the Gibbs point process.

Let \(\mu_d\) denote the percolation threshold of the \(d\)-dimensional Poisson Boolean model with radius \(1/2\):
\begin{equation}\label{eq:def mu_d}
\mu_d = \sup \{\mu>0 \mid \text{\(L(\, \cdot \,)\) has no unbounded connected component \(P^{Poi(\mu)}\)-a.s.} \},
\end{equation}
where \(L(\omega)= \bigcup_{x\in\omega}B(x,1/2)\). It is known that \(0<\mu_d<\infty\) when \(d\geq 2\), and \(\mu_1=\infty\) (see, e.g., \cite{MR}).

We introduce the following condition:
\begin{description}
\item[\phantomsection \textbf{(U3)}\label{cd:U3}] \(R^d e^{-b} < \mu_d\).
\end{description}

Under Conditions \hyperref[cd:U1]{\textbf{(U1)}}--\hyperref[cd:U3]{\textbf{(U3)}}, the corresponding Gibbs point process exists uniquely (see, e.g., \cite{HH, LO}). Its distribution \(P^{Gib}\) is \emph{stationary} (see e.g., \cite{DV}), i.e., \(P^{Gib}(\tau_x A)=P^{Gib}(A)\) for all \(x\in\mathbb{R}^d\) and \(A\in\mathcal{F}\), where \(\tau_x A=\{\tau_x\omega\mid\omega\in A\}\). It is also \emph{ergodic}, meaning that for any \(A\in\mathcal{F}\) such that \(P^{Gib}(A\ominus\tau_x A)=0\) for all \(x\in\mathbb{R}^d\), it holds that \(P^{Gib}(A)\in\{0,1\}\) (see \cite{Preston}), where \(A\ominus B\) denotes the symmetric difference of \(A\) and \(B\). As in the case of the Poisson point process, ergodicity implies that the IDS for the Gibbs point process is \(P^{Gib}\)-almost surely independent of the configuration \(\omega\) (see, e.g., \cite{CL, PF}).

Regarding the comparison with the Poisson point process, the Gibbs point process satisfies a stochastic domination property. We say that a measurable function \(f\) on \(\mathcal{C}\) is \emph{increasing} if \(f(\omega_1)\leq f(\omega_2)\) whenever \(\omega_1\subset\omega_2\). 
Under Conditions \hyperref[cd:U1]{\textbf{(U1)}} and \hyperref[cd:U2]{\textbf{(U2)}}, the measure \(P^{Gib}_{\Lambda_n,\gamma}\) is \emph{stochastically dominated} by \(P^{Poi(e^{-b})}_{\Lambda_n}\) for any \(n>0\) and \(\gamma\in\mathcal{C}\) (see \cite{GK}): for any bounded increasing measurable function \(f\) on \(\mathcal{C}_{\Lambda_n}\),
\begin{equation}\label{eq:bdd dom}
\int_{\mathcal{C}_{\Lambda_n}}f dP^{Gib}_{\Lambda_n,\gamma} \leq \int_{\mathcal{C}_{\Lambda_n}} f dP^{Poi(e^{-b})}_{\Lambda_n}.
\end{equation} 
The periodic measure \(P^{per}_{\Lambda_n}\) also satisfies this domination when \(n>2R\). Furthermore, by the DLR equation \eqref{eq:DLR} and \eqref{eq:bdd dom}, for any bounded increasing measurable function \(f\) on \(\mathcal{C}\), it holds that
\begin{equation}\label{eq:dom}
\begin{split}
\int_{\mathcal{C}}f(\omega_{\Lambda_n}) P^{Gib}(d\omega)&= \int_{\mathcal{C}}\int_{\mathcal{C}_{\Lambda_n}} f(\eta) P^{Gib}_{\Lambda_n,\gamma}(d\eta) P^{Gib}(d\gamma)\\
&\leq \int_{\mathcal{C}}f(\omega_{\Lambda_n}) P^{Poi(e^{-b})}(d\omega).
\end{split}
\end{equation}

Another essential tool in our analysis is the \emph{disagreement coupling}. We define the event \(\Omega_{\Lambda_n}= \{(\omega_1, \omega_2) \in \mathcal{C}_{\Lambda_n}^2 \mid (\omega_1)_{\Lambda_{n/2}}=(\omega_2)_{\Lambda_{n/2}}\}\). The existence of the disagreement coupling is summarized as follows.

\begin{proposition}\label{prop:dac}
Assume Conditions \hyperref[cd:U1]{\textbf{(U1)}}--\hyperref[cd:U3]{\textbf{(U3)}}. Fix \(\gamma_1, \gamma_2 \in\mathcal{C}\) and \(n>0\). There exists a probability measure \(P^{dac}_{\Lambda_n, \gamma_1, \gamma_2}\) on \((\mathcal{C}_{\Lambda_n}^2, \mathcal{F}_{\Lambda_n}^{\otimes 2})\) such that for each \(i=1,2\), the marginal distribution of \(\omega_i\) under \(P^{dac}_{\Lambda_n, \gamma_1, \gamma_2}\) is \(P^{Gib}_{\Lambda_n, \gamma_i}\), and
\begin{equation}\label{eq:dac}
P^{dac}_{\Lambda_n, \gamma_1, \gamma_2}(\mathcal{C}_{\Lambda_n}^2 \setminus \Omega_{\Lambda_n}) \leq Ce^{-n/C},
\end{equation}
where the constant \(C>0\) is independent of \(\gamma_1, \gamma_2\), and \(n\).
\end{proposition}

\begin{proof}
Let \((\omega_1, \omega_2, \omega_3)\) denote the canonical coordinates on \(\mathcal{C}_{\Lambda_n}^3\). 
By standard disagreement coupling arguments (see \cite{HH, LO}), there exists a coupling \(P\) on \(\mathcal{C}_{\Lambda_n}^3\) such that the marginal of \(\omega_i\) is \(P^{Gib}_{\Lambda_n, \gamma_i}\) for \(i=1,2\), and the marginal of \(\omega_3\) is \(P^{Poi(e^{-b})}_{\Lambda_n}\). Furthermore, \(P\)-a.s., for any \(y\in \omega_1 \ominus \omega_2\), there exists \(z\in\gamma_1\cup\gamma_2\) and a subset \(\eta\subset \omega_3 \) such that \(\bigcup _{x\in\eta\cup\{y,z\}}B(x,R/2)\) is connected.
We define \(P^{dac}_{\Lambda_n, \gamma_1, \gamma_2}\) on \(\mathcal{C}_{\Lambda_n}^2\) as the marginal distribution of the first two coordinates \((\omega_1, \omega_2)\) of \(P\). By standard percolation and scaling arguments, under Condition \hyperref[cd:U3]{\textbf{(U3)}}, the probability that a connected component of \(\bigcup_{x \in \omega_3} B(x,R/2)\) intersects both \((\Lambda_{n-2R})^c\) and \(\Lambda_{n/2+2R}\) decays exponentially in \(n\) (see, e.g., \cite{MR}). This yields \(P^{dac}_{\Lambda_n, \gamma_1, \gamma_2}(\mathcal{C}_{\Lambda_n}^2 \setminus \Omega_{\Lambda_n}) \leq C e^{-n/C}\) for some constant \(C>0\).
\end{proof}

\section{Periodic approximation}\label{sec:PA}
In this section, we extend the periodic approximation method to the case where \(\Gamma\) is a Gibbs point process with distribution \(P^{Gib}\).
Let \(\tilde{\pi}_n : \mathcal{C}\to\mathcal{C}\) be the map defined by \(\tilde{\pi}_n(\omega) = \pi_n (\omega_{\Lambda_n})\).
In \cite{KP}, for Poisson point processes, the periodic approximation scheme was established by using the fact that the probability measure \(P^{Poi} \circ\tilde{\pi}_n^{-1}\) on \(\mathcal{C}\) is stationary under translations. However, for a general Gibbs point process, the measure \(P^{Gib}\circ\tilde{\pi}^{-1}_{n}\) is not necessarily stationary. To address this difficulty, we use the periodic Gibbs measure \(P^{per}_{\Lambda_n}\) instead of \(P^{Gib}\circ\tilde{\pi}^{-1}_{n}\).

Throughout this paper, we consider the operator \(H_{\omega}\) introduced in Section \ref{sec:intro} with the single-site potential \(V\) given by \(V=V_1+V_2\), where \(V_1\) and \(V_2\) satisfy the following conditions:
\begin{description}
\item[\textbf{(V1)}]\phantomsection \label{cd:V1} \(|V_1(x)|\leq C\exp(-|x|/C)\) for some constant \(C>0\);
\item[\textbf{(V2)}]\phantomsection \label{cd:V2} \(V_2\in L^{p}(\mathbb{R}^d)\) for some \(p>p(d)\), and \(V_2\) has compact support;
\item[\textbf{(V3)}]\phantomsection \label{cd:V3} \(\essinf V<0\).
\end{description}
Here, \(p(d)\) is defined by
\begin{equation}\label{eq:p(d)}
p(d)=
\begin{cases}
2 & \text{if \(d\leq 3\)},\\
d/2 & \text{if \(d\geq 4\)}.
\end{cases}
\end{equation}

\begin{remark}
The values of \(p(d)\) in Condition \hyperref[cd:V2]{\textbf{(V2)}} are chosen to satisfy the requirements of the general theory of random Schr\"{o}dinger operators. While the original formulation in \cite{KP} allowed for a slightly weaker condition at \(d=3\), the standard regularity requirements presented in the monographs \cite{CL} and \cite{PF} typically assume \(p > 2\) for \(d \leq 3\) to establish the basic properties of the IDS. By adopting this convention, we can directly rely on the general framework provided in these foundational works.
\end{remark}

For every \(n\in\mathbb{Z}_{>0}\) and \(\omega\in\mathcal{C}_{\Lambda_n}\), we define the operator \(H_{\omega,n}\) on \(L^2(\mathbb{R}^d)\) by
\begin{equation}
H_{\omega,n} = -\Delta + V_{\omega,n},\quad \text{where}\quad V_{\omega,n}= \sum_{x\in \pi_n (\omega)} V(\,\cdot-x).
\end{equation}
Let \(N_{\omega, n}(E)\) denote the IDS of the operator \(H_{\omega, n}\), defined in the same manner as \(N(E)\).

Conditions \hyperref[cd:V1]{\textbf{(V1)}} and \hyperref[cd:V2]{\textbf{(V2)}} ensure that the operators \(H_{\omega}\) and \(H_{\omega,n}\) are well-defined as self-adjoint operators \(P^{Poi}\)-almost surely and \(P^{Gib}\)-almost surely (see e.g., \cite{CL, Kirsch, PF, RS2}), provided that the energy function \(U\) satisfies Conditions \hyperref[cd:U1]{\textbf{(U1)}}--\hyperref[cd:U3]{\textbf{(U3)}}.

Our main result in this section is the following theorem, which is a generalization of Lemma 2.3 in \cite{KP}.
\begin{theorem}[Periodic approximation]\label{thm:per app}
Suppose that the single-site potential \(V\) and the energy function \(U\) satisfy Conditions \hyperref[cd:V1]{\textbf{(V1)}}--\hyperref[cd:V3]{\textbf{(V3)}} and \hyperref[cd:U1]{\textbf{(U1)}}--\hyperref[cd:U3]{\textbf{(U3)}}, respectively.
For any \(\nu>0\), there exist constants \(\beta>0\) and \(E_{\nu}>0\) such that for any \(E\geq E_{\nu}\) and any \(n\in\mathbb{Z}_{>0}\) satisfying \(n\geq E^{\beta}\), it holds that
\begin{equation}
\int_{\mathcal{C}_{\Lambda_n}}\!\! N_{\omega,n}(-E-1) \, P^{per}_{\Lambda_n}(d\omega) - e^{-E^{\nu}} 
\leq N(-E)
\leq \int_{\mathcal{C}_{\Lambda_n}}\!\! N_{\omega,n}(-E+1) \, P^{per}_{\Lambda_n}(d\omega) + e^{-E^{\nu}}.
\end{equation}
\end{theorem}

While the overall structure of the proof is analogous to that of \cite[Lemmas 2.1 and 2.3]{KP}, a key modification lies in the estimation of the difference between \(P^{Gib}_{\Lambda_n,\gamma_1}\) and \(P^{Gib}_{\Lambda_n,\gamma_2}\) (\(\gamma_1\neq \gamma_2\)) by the disagreement coupling (see Section \ref{subsec:properties}).
To avoid redundancy, we follow the steps of the original proofs and focus on the essential changes.

\begin{proof}
Throughout the proof, \(C>1\) denotes a generic sufficiently large constant that may change from line to line.

Let \(C_{0}^{\infty}(\mathbb{R})\) denote the set of all compactly supported smooth functions on \(\mathbb{R}\). For any \(\phi\in C_{0}^{\infty}(\mathbb{R})\), by the stationarity of the Gibbs point process and \cite[Theorem 1.3 and Remark 2.3]{KP}, it holds that
\begin{equation}\label{eq:tr Gib}
\int_{\mathbb{R}}\phi(x) N(dx) = \int_{\mathcal{C}} \tr(\mathbb{1}_{\Lambda_1} \phi(H_{\omega})\mathbb{1}_{\Lambda_1}) P^{Gib}(d\omega),
\end{equation}
where \(\tr(\,\cdot\,)\) denotes the trace of an operator.
By \eqref{eq:per GPP stationary} and the argument in \cite{KP}, we obtain
\begin{equation}\label{eq:tr per}
\int_{\mathcal{C}_{\Lambda_n}} \int_{\mathbb{R}}\phi(x) N_{\omega,n}(dx) P^{per}_{\Lambda_n}(d\omega) = \int_{\mathcal{C}_{\Lambda_n}} \tr(\mathbb{1}_{\Lambda_1} \phi(H_{\omega,n})\mathbb{1}_{\Lambda_1}) P^{per}_{\Lambda_n}(d\omega).
\end{equation}
We set \(f(\omega)=\tr(\mathbb{1}_{\Lambda_1}\phi(H_{\omega,n})\mathbb{1}_{\Lambda_1})\) and \(g(\omega)=\tr(\mathbb{1}_{\Lambda_1}\phi(H_{\omega})\mathbb{1}_{\Lambda_1})\).

For simplicity, we set \(m = n/2\). By \eqref{eq:tr Gib}, \eqref{eq:tr per}, and the DLR equation \eqref{eq:per DLR}, for all sufficiently large \(n>0\), we have
\begin{equation}
\begin{split}
&\left|\int_{\mathcal{C}_{\Lambda_n}} \int_{\mathbb{R}}\phi(x) N_{\omega,n}(dx) P^{per}_{\Lambda_n}(d\omega) -\int_{\mathbb{R}}\phi(x) N(dx)\right| \\
&= \left|\int_{\mathcal{C}_{\Lambda_n}}f(\omega)P_{\Lambda_n}^{per}(d\omega)-\int_{\mathcal{C}}g(\omega)P^{Gib}(d\omega)\right|\\
&= \Bigg|\int_{\mathcal{C}_{\Lambda_n}} \int_{\mathcal{C}_{\Lambda_{m}}}f(\omega_1\cup(\gamma_1)_{\Lambda_{m}^c})P^{Gib}_{\Lambda_{m},\gamma_1}(d\omega_1) \, P^{per}_{\Lambda_n}(d\gamma_1)\\
&\qquad\qquad\qquad\qquad\qquad\qquad- \int_{\mathcal{C}} \int_{\mathcal{C}_{\Lambda_{m}}}g(\omega_2\cup(\gamma_2)_{\Lambda_{m}^c})P^{Gib}_{\Lambda_{m},\gamma_2}(d\omega_2) \, P^{Gib}(d\gamma_2)\Bigg| \\
&\leq \int_{\mathcal{C}}\int_{\mathcal{C}_{\Lambda_n}} \left( \int_{(\mathcal{C}_{\Lambda_{m}})^2} \left| h(\omega_1,\omega_2,\gamma_1,\gamma_2)\right| P^{dac}_{\Lambda_{m},\gamma_1,\gamma_2}(d(\omega_1,\omega_2))\right)P^{per}_{\Lambda_n}(d\gamma_1) \, P^{Gib}(d\gamma_2)\\
&= L_1+L_2,
\end{split}
\end{equation}
where we have used the disagreement coupling introduced in Proposition \ref{prop:dac}.
Here, we set \(h(\omega_1,\omega_2,\gamma_1,\gamma_2)=f(\omega_1\cup(\gamma_1)_{\Lambda_{m}^c})-g(\omega_2\cup(\gamma_2)_{\Lambda_{m}^c})\), and \(L_1\) is the contribution from the event \(\Omega_{\Lambda_{m}}\), given by
\begin{equation}\label{eq:L_1}
L_1= \iint \left( \int_{\Omega_{\Lambda_{m}}} \left|h(\omega_1,\omega_2,\gamma_1,\gamma_2) \right| dP^{dac}_{\Lambda_{m},\gamma_1,\gamma_2} \right)dP^{per}_{\Lambda_n} \, dP^{Gib},
\end{equation}
while \(L_2\) is defined analogously with the inner integral taken over \((\Omega_{\Lambda_{m}})^c\).

We pick \(\alpha\in(0,1)\), \(p'\in(p(d),p)\), and \(r_0>0\) such that \(\supp V_2 \subset \Lambda_{r_0}\), where \(\supp V_2\) is the closure of \(\{x\in\mathbb{R}^d \mid V_2 (x)\neq 0\}\).
For every \(k\in\mathbb{Z}_{\geq 0}\), we define
\begin{equation}
\Omega_k =\{\omega\in\mathcal{C}\mid m(\omega,\Lambda_{r_0}(x))<k(|x|^{(1-p'/p)\alpha}+1)\ \text{for all \(x\in r_0 \mathbb{Z}^d\)}\}.
\end{equation}

By \eqref{eq:bdd dom}, \eqref{eq:dom}, and the Appendix in \cite{KP}, for any \(k, l\in\mathbb{Z}_{>0}\), \(n>0\) such that \(n>2R\), and \(\gamma\in\mathcal{C}\), the following estimates hold:
\begin{equation}\label{eq:prob Omega_k}
P^{Gib}_{\Lambda_n ,\gamma}(\mathcal{C}_{\Lambda_n}\setminus\Omega_k),\ 
P^{per}_{\Lambda_n}(\mathcal{C}_{\Lambda_n}\setminus\Omega_k),\ 
P^{Gib}(\Omega_k^c)\leq C \frac{(e^{-b}r_0^d)^k}{k!},
\end{equation}
and
\begin{equation}\label{eq:prob dac Omega_k}
P^{dac}_{\Lambda_{n},\gamma_1,\gamma_2}(\omega_1\in\Omega_{k}\setminus\Omega_{k-1},\ \omega_2\in\Omega_{l}\setminus\Omega_{l-1})\leq C\frac{(e^{-b}r_0^d)^{\max\{k,l\}}}{(\max\{k,l\})!}.
\end{equation}

We now estimate \(L_1\).
Fix an integer \(q>d/2\) and let \(\tilde{\phi}\) denote an \emph{almost analytic extension} of the function \((i+x)^q \phi(x)\) (see \cite{KP} for its properties). 
Let \(k_1, k_2, l_1, l_2 \in \mathbb{Z}_{> 0}\) and set \(K = \max\{k_1, k_2, l_1, l_2\}\). Consider any configurations \(\gamma_1 \in \Omega_{l_1} \setminus \Omega_{l_1-1}\), \(\gamma_2 \in \Omega_{l_2} \setminus \Omega_{l_2-1}\), \(\omega_1 \in \Omega_{k_1} \setminus \Omega_{k_1 -1}\), and \(\omega_2 \in \Omega_{k_2}\setminus \Omega_{k_2 -1}\). From \cite{KP}, for any \(x\in \mathbb{R}^d\), it follows that
\begin{equation}
\left\|\mathbb{1}_{\Lambda_1 (x)}\left(V_{\omega_1\cup(\gamma_1)_{\Lambda_m^c},n}-V_{\omega_2\cup (\gamma_2)_{\Lambda_m^c}}\right)\mathbb{1}_{\Lambda_1 (x)}\right\|_{\mathcal{L}(H^1,H^{-1})}\leq CK^{p/(p-p')}(1+|x|)^{\alpha},
\end{equation}
where \(\|\cdot \|_{\mathcal{L}(H^1, H^{-1})}\) denotes the operator norm from the Sobolev space \(H^1(\mathbb{R}^d)\) to its dual space \(H^{-1}(\mathbb{R}^d)\).
By \cite{KP}, in addition, if \((\omega_1, \omega_2) \in \Omega_{\Lambda_m}\), that is, \((\omega_1)_{\Lambda_{m/2}} = (\omega_2)_{\Lambda_{m/2}}\), it follows that for any \(x \in \Lambda_{m/4}\),
\begin{equation}
\left\|\mathbb{1}_{\Lambda_1 (x)}\left(V_{\omega_1\cup(\gamma_1)_{\Lambda_m^c},n}-V_{\omega_2\cup (\gamma_2)_{\Lambda_m^c}}\right)\mathbb{1}_{\Lambda_1 (x)}\right\|_{\mathcal{L}(H^1,H^{-1})}\leq CK e^{-n/C}.
\end{equation}

Following the argument in \cite{KP} and using these bounds, we can estimate \(|h(\omega_1, \omega_2, \gamma_1, \gamma_2)|\) as follows:
\begin{equation}\label{eq:h estimate}
\begin{split}
|h(\omega_1,\omega_2,\gamma_1,\gamma_2)| 
&\leq \frac{1}{2\pi}\int_{\mathbb{C}}\left|\frac{\partial\tilde{\phi}}{\partial \overline{z}}\right|
T(z, \omega_1, \omega_2, \gamma_1, \gamma_2)
\, dx\, dy\\
&\leq C\int_{\mathbb{C}}\left|\frac{\partial\tilde{\phi}}{\partial \overline{z}}\right|\,\bigg(\frac{n}{\eta(z,K)}\bigg)^{C} \exp\!\left(-\frac{n^{1/C}}{C}\eta(z,K)^C\right)dx\, dy,
\end{split}
\end{equation}
for sufficiently large \(n>0\).
Here, we set \(z=x+iy\),
\begin{equation}
\eta(z,K)= \frac{|\Im z|}{|z|+CK^{\frac{p}{p-p'}}+C},
\end{equation}
and
\begin{equation}
\begin{split}
&T(z, \omega_1, \omega_2, \gamma_1, \gamma_2)\\
&=\bigg\|\mathbb{1}_{\Lambda_1}\bigg( \left(i+H_{\omega_1\cup(\gamma_1)_{\Lambda_{m}^c},n}\right)^{-q} \left(z-H_{\omega_1\cup(\gamma_1)_{\Lambda_{m}^c},n}\right)^{-1}\\
&\quad\qquad\qquad\qquad\qquad\qquad- \left(i+H_{\omega_2\cup(\gamma_2)_{\Lambda_{m}^c}}\right)^{-q} \left(z-H_{\omega_2\cup(\gamma_2)_{\Lambda_{m}^c}}\right)^{-1} \bigg)\mathbb{1}_{\Lambda_1}\bigg\|_{\mathrm{tr}},
\end{split}
\end{equation}
where \(\|\cdot\|_{\mathrm{tr}}\) denotes the trace norm.
By substituting \eqref{eq:h estimate} into \eqref{eq:L_1}, taking the summation over \(k_1,k_2,l_1,l_2\), and applying the bounds in \eqref{eq:prob Omega_k} and \eqref{eq:prob dac Omega_k}, we obtain
\begin{equation}\label{eq:L1 estimate}
L_1 \leq C n^{-l/C}e^{Cl\log l} \int_{\mathbb{C}} \bigg|\frac{\partial\tilde{\phi}}{\partial\overline{z}}\bigg| \left(\frac{|x|+C}{|y|}\right)^{C+l} dx\, dy
\end{equation}
for any \(l\in\mathbb{Z}_{>0}\). Here, we used the fact that \(\tilde{\phi}(z)=0\) if \(|\Im z|\geq 1\). The details of the calculation are provided in \cite{KP}.

To estimate \(L_2\), we use Fubini's theorem and the Cauchy-Schwarz inequality to obtain
\begin{equation}
\begin{split}
L_2
&\leq \frac{1}{2\pi} \iint \left( P^{dac}_{\Lambda_m, \gamma_1, \gamma_2} ((\mathcal{C}_{\Lambda_m})^2 \setminus \Omega_{\Lambda_m}) \right)^{1/2} \\
&\qquad\qquad \times \left( \int_{\mathbb{C}} 
\bigg|\frac{\partial \tilde{\phi}}{\partial \overline{z}}\bigg| 
\Bigg( \int_{(\mathcal{C}_{\Lambda_m})^2 } 
T^2 \, dP^{dac}_{\Lambda_m, \gamma_1, \gamma_2}\Bigg)^{1/2} dx\, dy \right) dP^{per}_{\Lambda_n}\,  dP^{Gib}.
\end{split}
\end{equation}
Fix \(l_1,l_2,k_1,k_2\in\mathbb{Z}_{>0}\) and any configurations \(\gamma_1\in\Omega_{l_1}\setminus\Omega_{l_1-1}\), \(\gamma_2\in \Omega_{l_2}\setminus\Omega_{l_2-1}\), and \((\omega_1,\omega_2)\in(\mathcal{C}_{\Lambda_m})^2\) satisfying \(\omega_j \in \Omega_{k_j} \setminus \Omega_{k_j-1}\) for \(j=1,2\).
Following an argument similar to that in \cite{KP}, we obtain
\begin{equation}
T(z, \omega_1, \omega_2, \gamma_1, \gamma_2)^2
\leq C\eta(z,K)^{-C} 
\end{equation}

Similarly to the estimation of \(L_1\), we obtain
\begin{equation}\label{eq:L2 estimate}
L_2 \leq C e^{-n/C}\int_{\mathbb{C}} \bigg|\frac{\partial\tilde{\phi}}{\partial\overline{z}}(z)\bigg| \left(\frac{|x|+C}{|y|}\right)^{C} dx\, dy,
\end{equation}
where we used the bound \eqref{eq:dac}.

Combining \eqref{eq:L1 estimate}, \eqref{eq:L2 estimate}, and the estimation in \cite{KP} concerning the almost analytic extension \(\tilde{\phi}\), we obtain
\begin{equation}\label{eq:estimate int}
\begin{split}
&\left|\int_{\mathcal{C}_{\Lambda_n}} \int_{\mathbb{R}}\phi(x) N_{\omega,n}(dx) P^{per}_{\Lambda_n}(d\omega) -\int_{\mathbb{R}}\phi(x) N(dx)\right|\\
&\leq C'e^{C'l\log l}\left(n^{-l/C'}+e^{-n/C'}\right)\sup_{\substack{0\leq j< l+C' \\ x\in \mathbb{R}}} \left|(|x|+C')^{l+C'}\frac{d^j}{dx^j}\phi(x)\right|
\end{split}
\end{equation}
for all \(l,n\in\mathbb{Z}_{>0}\) and some constant \(C'>0\) independent of \(l,n\) and \(\phi\).
Finally, Theorem \ref{thm:per app} follows from \eqref{eq:estimate int} by an argument similar to that in the proof of Lemma 2.3 in \cite{KP}.
\end{proof}

\section{Weak interaction}\label{sec:weak}
In this section, we consider a Gibbs point process with an energy function \(U\) satisfying the following condition:
\begin{description}
\item [\phantomsection \textbf{(W)}\label{cd:W}]
there exists \(\alpha>0\) such that
\begin{equation}\label{eq:weak}
\varlimsup_{k\to+\infty} \frac{1}{k\log k}\,\sup_{\omega\in \Omega_k^{\alpha}} U(\omega) \leq 0,
\end{equation}
where
\begin{equation}
\Omega_k^{\alpha}= \{\omega\in\mathcal{C}\mid \#\omega =\lfloor k\rfloor,\ \omega\subset \Lambda_{k^{-\alpha}}\}.
\end{equation}
\end{description}
Here, \(\lfloor k \rfloor\) denotes the greatest integer less than or equal to \(k\).
Condition \hyperref[cd:W]{\textbf{(W)}} implies that the interaction among many points in a small region is sufficiently weak.
Let \(\tilde{\mathcal{A}}\) denote the set of all \(\alpha>0\) for which \eqref{eq:weak} holds, and set \(\tilde{\alpha}=\inf \tilde{\mathcal{A}}\).

To state the asymptotic behavior of \(\log N(-E)\) as \(E\to+\infty\), we introduce some notation.
By \cite{CL, RS4}, for a potential \(V\) satisfying Conditions \hyperref[cd:V1]{\textbf{(V1)}}--\hyperref[cd:V3]{\textbf{(V3)}}, for all sufficiently large \(g>0\), the infimum of the spectrum of \(H(g)=-\Delta+gV\) is a discrete eigenvalue. We denote this eigenvalue by \(E(g)\).
For convenience, we set \(E_{-}(g)=-E(g)\). The variational principle and a simple calculation show that \(E_{-}\) is a strictly increasing positive function on \((g_0, +\infty)\) for some sufficiently large \(g_0>0\). Thus, we can define \(g(E)\) as the inverse function of \(E_{-}\) for sufficiently large \(E>0\).

An \emph{asymptotic ground state} of \(H(g)\) is a map \((1,+\infty)\ni g \mapsto \psi_g \in H^{1}(\mathbb{R}^d)\) satisfying the following conditions:
\begin{itemize}
\item \(\|\psi_g\|=1\) for all \(g>1\);
\item there exist \(g_0>1\) and \(l_0>0\) such that \(\supp \psi_g \subset \Lambda_{l_0}\) for all \(g\geq g_0\);
\item the following limit holds:
\begin{equation}
\lim_{g\to+\infty}\frac{\langle (H(g)-E(g))\psi_g, \psi_g \rangle}{E(g)}=0.
\end{equation}
\end{itemize}
Here, \(\langle \cdot, \cdot \rangle\) and \(\|\cdot\|\) denote the standard inner product and the corresponding norm on \(L^2(\mathbb{R}^d)\), respectively.
Let \(\mathcal{V}\) denote the set of all asymptotic ground states; see \cite{KP} for their detailed properties. 
We define the constant \(\alpha^{*}\) by
\begin{equation}
\alpha^{*} = \inf_{\psi_g \in \mathcal{V}} \left( \inf \bigg\{ \alpha>0 \;\bigg\vert\; \lim_{g\to +\infty} \sup_{|y|\leq g^{-\alpha}}\left|\frac{g\langle(\tau_y V-V)\psi_g, \psi_g\rangle}{E(g)}\right|=0\bigg\}\right),
\end{equation}
where for \(y\in\mathbb{R}^d\), we define \(\tau_y V(x)=V(x-y)\).

Before stating our main results for Gibbs point processes, let us recall the results for the case of Poisson point processes.
If \(V\) satisfies Conditions \hyperref[cd:V1]{\textbf{(V1)}}--\hyperref[cd:V3]{\textbf{(V3)}}, and \(\Gamma\) is a Poisson point process with intensity \(\mu\), it is known that the IDS satisfies
\begin{equation}\label{eq:KP}
\begin{split}
&\log N(-E)\geq -(1+d\alpha^{*})g(E)\log g(E) (1+o(1)),\\
&\log N(-E)\leq -g(E)\log g(E) (1+o(1)),
\end{split}
\end{equation}
as \(E\to+\infty\).
These estimates were established in \cite{KP} using the periodic approximation method, which is also the primary tool in this paper.

In this section, we establish the following theorem.
\begin{theorem}\label{thm:weak}
Let the single-site potential \(V\) satisfy Conditions \hyperref[cd:V1]{\textbf{(V1)}}--\hyperref[cd:V3]{\textbf{(V3)}}.
Suppose that \(\Gamma\) is a Gibbs point process for an energy function \(U\) satisfying Conditions \hyperref[cd:U1]{\textbf{(U1)}}--\hyperref[cd:U3]{\textbf{(U3)}} and \hyperref[cd:W]{\textbf{(W)}}. Then, the corresponding IDS satisfies
\begin{align}
\log N(-E)& \geq -(1+d\max\{\tilde{\alpha},\alpha^{*}\})g(E)\log g(E)(1+o(1)),\label{eq:weak lower}\\
\log N(-E)& \leq -g(E)\log g(E)(1+o(1)),\label{eq:weak upper}
\end{align}
as \(E\to+\infty\).
In particular, if \(\tilde{\alpha}\leq \alpha^{*}\), then the IDS satisfies \eqref{eq:KP}.
\end{theorem}
Note that a Poisson point process with intensity \(\mu>0\) can be regarded as the Gibbs point process with the energy function \(U(\omega)=-\#\omega \log \mu\). In this sense, Theorem \ref{thm:weak} can be viewed as a generalization of the result for Poisson point processes established in \eqref{eq:KP}.

The following are examples of Gibbs point processes that satisfy the assumptions of Theorem \ref{thm:weak} and the condition \(\tilde{\alpha}\leq \alpha^{*}\). 
\begin{example}[Bounded local energy function]
Suppose that an energy function \(U\) satisfies Conditions \hyperref[cd:U1]{\textbf{(U1)}}--\hyperref[cd:U3]{\textbf{(U3)}}. If its local energy function is bounded, then Condition \hyperref[cd:W]{\textbf{(W)}} is satisfied and \(\tilde{\alpha}=0\). Hence, the condition \(\tilde{\alpha}\leq \alpha^{*}\) holds regardless of the value of \(\alpha^{*}\).
An example of such an energy function is the \emph{area energy function}, defined by
\begin{equation}
U(\omega)=\left| \bigcup_{x\in \omega} B(x,R/2) \right|,
\end{equation}
where \(0<R< \mu_{d}^{1/d}\) (\(\mu_d\) is defined in Section \ref{subsec:properties}).
Its local energy function \(u\) satisfies \(0\leq u\leq |B(0,R/2)|\), and its interaction function is given by
\begin{equation}
\Phi(\omega)= (-1)^{\# \omega+1} \left| \bigcap_{x\in \omega}B(x,R/2) \right|.
\end{equation}
\end{example}

\begin{example}[Pairwise energy function]\label{ex:pairwise}
An energy function \(U\) of the form
\begin{equation}\label{eq:pairwise}
U(\omega)= \frac{1}{2}\,\sum_{\substack{x,y\in\omega \\ x\neq y}}\varphi(x-y) -\# \omega \log z
\end{equation}
is called a \emph{pairwise energy function}. Here, \(\varphi:\mathbb{R}^d \to [0,\infty) \cup\{+\infty\}\) is a compactly supported measurable function that is symmetric (i.e., \(\varphi(x)=\varphi(-x)\) for all \(x\)), and \(z>0\) is a constant referred to as the \emph{activity}.
If \(\supp  \varphi \subset B(0,R)\) and \(0<zR^d< \mu_{d}\), and if \(\varphi\) satisfies
\begin{equation}
\varphi(x)=
\begin{cases}
O(|x|^{1/\alpha^{*}})\quad(x\to 0), & \text{if \(\alpha^{*}>0\)},\\
O(\exp(-|x|^{-t}))\quad(x\to 0), & \text{for some \(t>0\) if \(\alpha^{*}=0\)},
\end{cases}
\end{equation}
then this energy function satisfies Conditions \hyperref[cd:U1]{\textbf{(U1)}}--\hyperref[cd:U3]{\textbf{(U3)}}, \hyperref[cd:W]{\textbf{(W)}}, and the condition \(\tilde{\alpha}\leq \alpha^{*}\).
\end{example}

We next consider the asymptotic behavior of the IDS in the case where the single-site potential \(V\) satisfies Condition \hyperref[cd:V1]{\textbf{(V1)}} and the following condition:
\begin{description}
\item[\phantomsection \textbf{(V4)}\label{cd:V4}] For some \(q\in\mathbb{Z}_{>0}\), the function \(V_2\) is of the form
\begin{equation}
V_2= \sum_{i=1}^{q} \tau_{y_i}(W_i v_i),
\end{equation}
where each \(v_i\) is defined by
\begin{equation}
v_i(x)=h_i \bigg(\frac{x}{|x|} \bigg) |x|^{-\nu_i}\quad (i=1,\ldots, q).
\end{equation}
Furthermore, we assume that \(\min_{1\leq i\leq q} E_{v_i}<0\), where \(E_{v_i}\) denotes the minimum eigenvalue of the operator \(-\Delta+v_i\).
\end{description}
Here, \(y_1,\ldots, y_q\) are distinct points, the functions \(h_i\) are continuous on the unit sphere \(\mathbb{S}^{d-1}\), and \(W_i\) are compactly supported continuous functions on \(\mathbb{R}^d\) such that \(W_i(0)=1\). Moreover, the exponents \(\nu_i\) satisfy \(0 < \nu_i < \min\{d/2, 2\}\).
Note that \(-\Delta + g v_i\) \((g>0)\) are well-defined self-adjoint operators, as is the operator \(H(g)\).
We define the following constants:
\begin{gather}
\nu^{\dag} = \max \{ \nu_i \mid E_{v_i} < 0 \},\\
E_{0} = \min \{ E_{v_i} \mid \nu_i = \nu^{\dag} \}.
\end{gather}
When \(\Gamma\) is a Poisson point process and the single-site potential \(V\) satisfies Conditions \hyperref[cd:V1]{\textbf{(V1)}} and \hyperref[cd:V4]{\textbf{(V4)}}, it is known that
\begin{equation}\label{eq:weak V4 IDS}
\log N(-E) \sim -\left( 1+\frac{d-\nu^{\dag}}{2} \right) \left( \frac{E}{|E_0|} \right)^{1-\nu^{\dag} /2} \log E \quad (E\to +\infty).
\end{equation}
This result was proved in \cite{KP}.
From \cite{KP}, we remark that for this single-site potential, it follows that
\begin{equation}\label{eq:V4 alpha}
\alpha^{*}=\frac{1}{2-\nu^{\dag}}.
\end{equation}

Under the same assumptions on \(\Gamma\) as in Theorem \ref{thm:weak}, we obtain the following asymptotic behavior.
\begin{theorem}\label{thm:weak V4}
Suppose that a Gibbs point process \(\Gamma\) satisfies 
Conditions \hyperref[cd:U1]{\textbf{(U1)}}--\hyperref[cd:U3]{\textbf{(U3)}}, and \hyperref[cd:W]{\textbf{(W)}}. Furthermore, we assume that the single-site potential \(V\) satisfies Conditions \hyperref[cd:V1]{\textbf{(V1)}} and \hyperref[cd:V4]{\textbf{(V4)}}. Then, in addition to the lower estimate \eqref{eq:weak lower}, the following upper estimate holds:
\begin{equation}
\log N(-E) \leq -\left( 1+\frac{d-\nu^{\dag}}{2} \right) \left( \frac{E}{|E_0|} \right)^{1-\nu^{\dag} /2} \log E \,(1+o(1))\quad (E\to +\infty).
\end{equation} 
In particular, if \(\tilde{\alpha}\leq 1/(2-\nu^{\dag})\) holds, the corresponding IDS \(N(E)\) satisfies the asymptotic behavior \eqref{eq:weak V4 IDS}.
\end{theorem}

We now prove Theorem \ref{thm:weak} and Theorem \ref{thm:weak V4}.

\begin{proof}[Proof of Theorem \ref{thm:weak}]
Choose sufficiently large \(\rho, \beta>1\), and fix \(0<\veps<1\) and \(0<t<1\). We set \(n=2\lfloor E^{\beta}\rfloor \lfloor |\log E|^{\beta}\rfloor\), \(l=\lfloor |\log E|^{\beta} \rfloor \), and \(\alpha=\max\{\tilde{\alpha}, \alpha^{*}\}\).
To simplify the notation, we set \(\Lambda_{k,\veps} = \Lambda_{k^{-(1+\veps)\alpha}}\).
We define
\begin{equation}
\Omega^{1}_{k,E}= \{ \omega\in\mathcal{C} \mid m(\omega,\Lambda_{k,\veps})=\lfloor k \rfloor,\ m(\omega,\Lambda_l \setminus \Lambda_{k,\veps})=0 \},
\end{equation}
and
\begin{equation}\label{eq:Omega2}
\Omega^2_E=\{\omega\in\mathcal{C}\mid \text{\(m(\omega,\Lambda_{r_0}(x))<E^{\rho}(|x|^t+1)\) for any \(x\in r_0\mathbb{Z}^d\)}\},
\end{equation}
where we choose \(r_0\) such that \(\supp V_2\subset \Lambda_{r_0}\).

Let us first consider the lower bound \eqref{eq:weak lower}.
Let \(\Omega_E=\Omega^1_{k,E}\cap\Omega^2_{E}\cap\mathcal{C}_{\Lambda_n}\).
We note that
\begin{equation}\label{eq:weak lower1}
P^{per}_{\Lambda_n}(\Omega_E) \geq P^{per}_{\Lambda_n}(\Omega^1_{k,E}\cap\mathcal{C}_{\Lambda_n})- P^{per}_{\Lambda_n}((\Omega^2_E)^c \cap \mathcal{C}_{\Lambda_n}).
\end{equation}
Since \(P^{Poi(e^{-b})}_{\Lambda_n}\) stochastically dominates \(P^{per}_{\Lambda_n}\), the estimate \cite[Eq.~(5.2)]{KP} implies that for some \(C>0\),
\begin{equation}\label{eq:weak lower2}
P^{per}_{\Lambda_n}((\Omega^2_E)^c \cap \mathcal{C}_{\Lambda_n})
\leq C \frac{(e^{-b} r_0^d)^{\lfloor E^{\rho}\rfloor}}{\lfloor E^{\rho}\rfloor ! }
\leq \frac{1}{\lfloor E^{\rho}/2\rfloor !}
\end{equation}
for sufficiently large \(E>0\).

From the DLR equation \eqref{eq:per DLR}, for large \(E>0\), we obtain
\begin{equation}\label{eq:weak Omega}
\begin{split}
&P^{per}_{\Lambda_n}(\Omega^1_{k,E}\cap\mathcal{C}_{\Lambda_n})\\
&= \int_{\mathcal{C}_{\Lambda_n}}\frac{1}{Z_{\Lambda_l, \gamma}} \int_{\Omega^{1}_{k,E}\cap\mathcal{C}_{\Lambda_l}}e^{-U(\eta)} \, P^{Poi}_{\Lambda_l}(d\eta) \, P^{per}_{\Lambda_n}(d\gamma)\\
&\geq \frac{|\Lambda_{k,\veps}|^{\lfloor k\rfloor}}{\lfloor k\rfloor !} \exp\Bigg(-|\Lambda_l| e^{-b} -\sup_{\substack{\# \omega = \lfloor k\rfloor \\ \omega\in\mathcal{C}_{\Lambda_{k,\veps}}}} U(\omega)\Bigg),
\end{split}
\end{equation}
where we have used the estimate
\begin{equation}
Z_{\Lambda_l, \gamma} \leq \exp(|\Lambda_l|(e^{-b}-1)),
\end{equation}
and the fact that \(U_{\Lambda_l,\gamma}(\omega)= U(\omega)\) holds for all \(\omega\in\Omega^1_{k,E}\cap\mathcal{C}_{\Lambda_n}\) since the distance between \(\Lambda_{k,\veps}\) and \(\Lambda_l^c\) is greater than the interaction range \(R\) for sufficiently large \(E>0\). For any \(\veps'>0\), Condition \hyperref[cd:W]{\textbf{(W)}} implies that for sufficiently large \(k>0\), we have
\begin{equation}
\sup_{\substack{\# \omega = \lfloor k\rfloor \\ \omega\in\mathcal{C}_{\Lambda_{k,\veps}}}} U(\omega)\leq \veps' k\log k.
\end{equation}
Combining this estimate and \eqref{eq:weak Omega}, and setting \(k=g(E)(1+\veps)\), we obtain
\begin{equation}\label{eq:weak lower3}
\log P^{per}_{\Lambda_n} (\Omega^1_{k,E}\cap\mathcal{C}_{\Lambda_n})
\geq -(1+\veps)^2 (1+d\alpha+\veps') g(E)\log g(E) 
\end{equation}
for sufficiently large \(E>0\), where we used the fact that \(g(E)>E^{\beta'}\) for some \(\beta'>0\) (see the Appendix in \cite{KP}).
From \eqref{eq:weak lower1}, \eqref{eq:weak lower2}, and \eqref{eq:weak lower3}, we obtain
\begin{equation}
\varliminf_{E\to+\infty}\frac{\log P^{per}_{\Lambda_n}(\Omega_E)}{g(E)\log g(E)}\geq -(1+\veps)^2 (1+d\alpha+\veps').
\end{equation}
Letting \(\veps, \veps' \to 0\), from the proof of Proposition 3.1 in \cite{KP}, as \(E\to+\infty\), it follows that
\begin{equation}
\log \int_{\mathcal{C}_{\Lambda_n}} N_{\omega, n} (-E-1) \, P^{per}_{\Lambda_n}(d\omega) \geq -(1+ d\alpha) g(E)\log g(E)(1+o(1)).
\end{equation}
Using Theorem \ref{thm:per app} and following an argument yielding Proposition 3.2 in \cite{KP}, we obtain the lower bound in \eqref{eq:weak lower}.

Now, we consider the upper bound \eqref{eq:weak upper}.
Fix \(0<\veps<1\) and \(\delta>0\). 
Since \(P^{per}_{\Lambda_n}\) is stochastically dominated by \(P^{Poi(e^{-b})}_{\Lambda_n}\), from \cite[Lemma 4.3]{KP}, it follows that for \(k\geq E^{\delta}\) and sufficiently large \(E>0\),
\begin{equation}
\begin{split}
&\log P^{per}_{\Lambda_n}(\{\omega\in\mathcal{C}_{\Lambda_n}\mid \text{there exists \(x\in\mathbb{R}^d\) such that \(m(\pi_n (\omega) , \Lambda_{2l}(x))>k\)}\})\\
&\leq -(1-\veps) k\log k.
\end{split}
\end{equation} 
By setting \(k=g(E-2)(1-\veps)\) and following the same argument as in the proof of Proposition 4.1 in \cite{KP}, we obtain
\begin{equation}
\log \int_{\mathcal{C}_{\Lambda_n}} N_{\omega,n}(-E+1) \, P^{per}_{\Lambda_n}(d\omega)
\leq -g(E)\log g(E)(1+o(1))\quad (E\to+\infty).
\end{equation}
Note that, by stochastic domination, an assertion similar to Lemma 4.1 in \cite{KP} holds also for \(P^{per}_{\Lambda_n}\).
We obtain the upper bound \eqref{eq:weak upper} from Theorem \ref{thm:per app} by an argument similar to the one used to derive Proposition 4.2 in \cite{KP}.
\end{proof}

\begin{proof}[Proof of Theorem \ref{thm:weak V4}]
Let us recall \eqref{eq:V4 alpha}.
Using \cite[Lemmas 4.4 and 4.5]{KP}, combined with the fact that \(P^{per}_{\Lambda_n}\) is stochastically dominated by \(P^{Poi(e^{-b})}_{\Lambda_n}\), we obtain the upper bound by following the same argument as in the proofs of Proposition 4.4 and Theorem 1.6 in \cite{KP}:
\begin{equation}
\log N(-E)\leq -\left( 1+\frac{d-\nu^{\dag}}{2} \right) \left( \frac{E}{|E_0|} \right)^{1-\nu^{\dag}/2}\log E\, (1+o(1))\quad (E\to+\infty).
\end{equation}
The lower bound follows immediately from Theorem \ref{thm:weak} and \eqref{eq:V4 alpha}.
\end{proof}

\section{Pairwise energy function}\label{sec:pairwise}
In this section, we consider a Gibbs point process with a pairwise energy function. After stating our main results in Section \ref{subsec:pairwise result}, we devote the remainder of this section to their proofs.

\subsection{Main results}\label{subsec:pairwise result}
We deal with a Gibbs point process with a pairwise energy function \(U\) defined by \eqref{eq:pairwise}.
Moreover, we assume that the single-site potential \(V\) satisfies the following condition:
\begin{description}
\item[\phantomsection \textbf{(V3')}\label{cd:V3'}] \(\essinf V=-\infty\).
\end{description}
Here, \(\varphi\) is as in \eqref{eq:pairwise}. We define the limit inferior and limit superior of \(\varphi\) by
\begin{equation}
\varliminf_{x\to y}\varphi(x)=\lim_{r\to 0}\, \inf_{0<|x-y|<r} \varphi(x),\qquad\varlimsup_{x\to y}\varphi(x)=\lim_{r\to 0} \, \sup_{0<|x-y|<r} \varphi(x).
\end{equation}
We introduce the following condition on the limit of \(\varphi\) at the origin: 
\begin{description}
\item[\phantomsection \textbf{(L)}\label{cd:L}] the limit \(a_0=\lim_{x\to0}\varphi(x)\) exists and satisfies \(0<a_0<+\infty\).
\end{description}

We set
\begin{equation}\label{eq:S_a}
S_a=\{x\in\mathbb{R}^d \mid \varphi(x)\geq a \}\cup \{0\}.
\end{equation}
Let \(A\) denote the set of all \(a>0\) such that \(\Int S_a\), the interior of \(S_a\), contains the origin. Note that \(A\) is a non-empty set whenever \(\varliminf_{x\to 0}\varphi(x)>0\).

For \(l,a>0\), let \(X_{l,a}\) be the maximum cardinality of a configuration \(\omega \subset \Lambda_l + \supp V_{2,-}\) such that \(x-y \notin S_a\) for all distinct \(x, y \in \omega\), and set \(T_a=\lim_{l\downarrow 0}X_{l,a}\),
where \(V_{2,-}=\min\{V_{2},0\}\).
Note that the limit \(T_a\) is well-defined since \(X_{l,a}\) is a non-decreasing positive integer-valued function of \(l\).

For every positive integer \(q\), we set
\begin{equation}
\tilde{D}_q = \left\{\Big( (c_j)_{j=1}^q, (x_j)_{j=1}^q \Big)\in \mathbb{R}_{>0}^q \times (\mathbb{R}^d)^{q} \,\middle\vert\, \sum_{j=1}^{q}c_j=1,\ x_i \neq x_j \ (1\leq i< j\leq q) \right\}.
\end{equation}
Then, we define the subset \(D_q\subset \tilde{D}_q\) by
\begin{equation}\label{eq:def D}
D_q = \left\{\Big((c_j)_{j=1}^q, (x_j)_{j=1}^q\Big)\in \tilde{D}_q \,\middle\vert \, \essinf \sum_{j=1}^{q}c_j \tau_{x_j}V<0\right\}.
\end{equation}
Under Conditions \hyperref[cd:V1]{\textbf{(V1)}}, \hyperref[cd:V2]{\textbf{(V2)}}, and \hyperref[cd:V3']{\textbf{(V3')}}, for each \((\bm{c}_q,\bm{x}_q) \in \tilde{D}_q\), let \(E_{\bm{c}_{q},\bm{x}_q}(g)\) denote the infimum of the spectrum of
\begin{equation}
-\Delta+g\sum_{j=1}^{q} c_j \tau_{x_j}V,
\end{equation}
where \(\bm{c}_{q}=(c_1,\ldots, c_q)\) and \(\bm{x}_{q}=(x_1,\ldots,x_q)\).
We note that this self-adjoint operator is well-defined in the same manner as \(H(g)\).
We set \(E_{-,\bm{c}_q,\bm{x}_q}(g)=-E_{\bm{c}_q,\bm{x}_q}(g)\).
When \((\bm{c}_q,\bm{x}_q)\in D_q\), as is the case for \(E_{-}(g)\), the function \(E_{-,\bm{c}_q,\bm{x}_q}(g)\) is positive and strictly increasing for sufficiently large \(g>0\), allowing us to define its inverse function \(g_{\bm{c}_q,\bm{x}_q}(E)\).
We remark that \(E_{-}(g)=E_{-,(1),(0)}(g)\) and \(g(E)=g_{(1),(0)}(E)\).
The following is the main theorem in this section.

\begin{theorem}\label{thm:pairwise}
Let \(\Gamma\) be a Gibbs point process for a pairwise energy function \(U\) of the form \eqref{eq:pairwise} satisfying Conditions \hyperref[cd:U1]{\textbf{(U1)}} and \hyperref[cd:U3]{\textbf{(U3)}}.
We assume that \(V\) satisfies Conditions \hyperref[cd:V1]{\textbf{(V1)}}, \hyperref[cd:V2]{\textbf{(V2)}}, and \hyperref[cd:V3']{\textbf{(V3')}}.
Then, the following results hold.
\begin{enumerate}
\item[\textbf{(a)}] For any positive integer \(q\) and any \((\bm{c}_q,\bm{x}_q)=((c_1,\ldots,c_q),(x_1,\ldots,x_q))\in D_q\),
\begin{equation}\label{eq:lower bound}
\varliminf_{E\to+\infty} \frac{\log N(-E)}{g_{\bm{c}_q,\bm{x}_q}(E)^2} \geq -\frac{1}{2}\sum_{1\leq i,j \leq q} c_i c_j \varlimsup_{x\to x_i-x_j}\varphi(x).
\end{equation}
In particular, taking \(q=1\), if \(\varlimsup_{x\to 0}\varphi(x)>0\), we have
\begin{equation}
\log N(-E) \geq -\frac{1}{2}\varlimsup_{x\to 0}\varphi(x)g(E)^2 (1+o(1))\quad (E\to +\infty);
\end{equation}
\item[\textbf{(b)}] If \(\varliminf_{x\to 0}\varphi(x)>0\), then
\begin{equation}\label{eq:upper bound}
\varlimsup_{E\to+\infty} \frac{\log N(-E)}{g(E)^2}\leq -\frac{1}{2} \sup_{a\in A}\frac{a}{T_a}.
\end{equation}
\end{enumerate}
Here, we set \(\log(0)=-\infty\) for convenience. 
\end{theorem}

Under some additional conditions, we can determine the leading term of \(\log N(-E)\):
\begin{proposition}\label{prop:pairwise}
Let \(\Gamma\) and \(V\) satisfy the conditions in Theorem \ref{thm:pairwise} and Condition \hyperref[cd:L]{\textbf{(L)}}. Then, we have the following results.
\begin{enumerate}
\item[\textbf{(a)}]
If \(x-y \in \Int S_a\) for all \(a\in (0,a_0)\) and all \(x,y\in \supp V_{2,-}\), then it holds that
\begin{equation}\label{eq:IDS pairwise}
\log N(-E) \sim -\frac{a_0}{2} g(E)^2\quad (E\to+\infty).
\end{equation}
\item[\textbf{(b)}] If \(V\) is bounded outside any neighborhood of some \(x_0\in\mathbb{R}^d\), then the corresponding IDS \(N(-E)\) satisfies \eqref{eq:IDS pairwise}.
\end{enumerate}
\end{proposition}
Unlike the weak interaction case where the asymptotic order is \(g(E) \log g(E)\), the order here is \(g(E)^2\). This reflects the fact that the total interaction among points in a small region is proportional to the square of the number of points.

We provide an example of a Gibbs point process that satisfies the conditions of Proposition \ref{prop:pairwise}.
\begin{example}[Strauss process]\label{ex:Strauss}
Suppose that the assumptions in Theorem \ref{thm:pairwise} hold.
Let the Gibbs point process \(\Gamma\) be a \emph{Strauss process} introduced in \cite{Strauss}, which is characterized by
\begin{equation}\label{eq:Strauss}
\varphi(x)=
\begin{cases}
a_0\ &(|x|\leq r),\\
0 &(|x|>r)
\end{cases}
\end{equation}
for some \(a_0 ,r>0\). 
If \(\diam (\supp V_{2,-}) <r\), then from Proposition \ref{prop:pairwise}, we obtain
\begin{equation}
\log N(-E)\sim -\frac{a_0}{2}g(E)^2\quad (E\to+\infty).
\end{equation}
\end{example}

Several remarks on the above results are in order.
\begin{remark}\label{rem:g}
For any \(q\in\mathbb{Z}_{>0}\), \((\bm{c}_q,\bm{x}_q)\in D_{q}\), and \(V\) satisfying Conditions \hyperref[cd:V1]{\textbf{(V1)}}, \hyperref[cd:V2]{\textbf{(V2)}}, and \hyperref[cd:V3']{\textbf{(V3')}}, there exists a constant \(C>0\) such that for sufficiently large \(E>0\),
\begin{equation}
g_{\bm{c}_q,\bm{x}_q}(E) \geq C g(E).
\end{equation}
This follows from Theorem \ref{thm:pairwise} \textbf{(a)} and \textbf{(b)} (by taking \(\varphi\) as in \eqref{eq:Strauss}).
\end{remark}

\begin{remark}
Under the assumptions in Theorem \ref{thm:pairwise}, when \(0<\varliminf_{x\to 0}\varphi(x)\leq \varlimsup_{x\to 0}\varphi(x)<+\infty\), we obtain 
\begin{equation}
-\beta g(E)^2 (1+o(1)) \leq \log N(-E) \leq -\frac{1}{2}\sup_{a\in A}\frac{a}{T_a} g(E)^2 (1+o(1))\quad (E\to+\infty),
\end{equation}
i.e., we know that \(\log N(-E)\asymp g(E)^2\ (E\to+\infty)\),
where
\begin{equation}
\beta= \frac{1}{2}\inf_{q\in\mathbb{Z}_{>0}}\, \inf_{(\bm{c}_q,\bm{x}_q)\in D_q}\left(\sum_{1\leq i,j\leq q}c_i c_j \varlimsup_{x\to x_i-x_j} \varphi(x) \varlimsup_{E\to+\infty}\left(\frac{g_{\bm{c}_q,\bm{x}_q}(E)}{g(E)}\right)^2\right).
\end{equation}
We note that \(0<\beta<+\infty\) (we know \(\beta>0\) from the inequality above and the fact that \(\sup_{a\in A}(a/T_a) > 0\)). 
\end{remark}

\begin{remark}
We briefly comment on the case of a hard-core interaction, where points are strictly prohibited from coming closer than a certain distance (i.e., there exists \(r_{hc}>0\) such that \(\varphi(x)=+\infty\) for \(|x|\leq r_{hc}\)).
In such a model, the number of points in any bounded region \(\Lambda\) is bounded by a constant depending only on \(\Lambda\) and \(r_{hc}\). Combining this with the exponential decay of \(V\) at infinity (Condition \hyperref[cd:V1]{\textbf{(V1)}}) and the fact that \(V\) is form-bounded with respect to \(-\Delta\) with relative bound zero, we find that the infinite sum \(\sum_{x\in\omega} V(\,\cdot-x)\) is form-bounded with relative bound zero uniformly in \(\omega\). Consequently, the spectrum of the operator \(H_{\omega}\) is contained in \([-M,+\infty)\) for some \(M>0\) independent of \(\omega\), which implies \(N(-E)=0\) for sufficiently large \(E>0\).
\end{remark}

\subsection{Proofs of the results in this section}
\begin{proof}[Proof of Theorem \ref{thm:pairwise} \textbf{(a)}]
We put \(n=\lfloor E\rfloor^{\beta}\) and \(l=|\log E|^{\beta}\), where \(\beta>0\) is sufficiently large.
Fix \(q\in\mathbb{Z}_{>0}\) and \((\bm{c}_q,\bm{x}_q)=((c_1,\ldots,c_q),(x_1,\ldots,x_q))\in D_q\).
From the proof of Lemma 5.6 in \cite{KP}, there exist \(\alpha>0\) and an asymptotic ground state \(\psi_g\) of \(-\Delta+g(\sum_{j=1}^q c_j \tau_{x_j}V)\) (defined analogously to that for \(H(g)\)) such that for any \(x\in\mathbb{R}^d\),
\begin{equation}\label{eq:ground state}
\lim_{g\to+\infty} \sup_{|y|\leq g^{-\alpha}} \left(\frac{g\left|\left\langle \left(\tau_{y+x} V-\tau_{x}V\right)\psi_g, \psi_g \right\rangle\right|}{E_{-,\bm{c}_q,\bm{x}_q}(g)} \right)=0.
\end{equation}
We define \(\beta^{+}_{j,k}=1\) and \(\beta^{-}_{j,k}=0\) if \(\langle V(\,\cdot-x_j)\psi_k,\psi_k\rangle\geq 0\), and \(\beta^{+}_{j,k}=2\) and \(\beta^{-}_{j,k}=1\) otherwise, to properly bound the energy from above regardless of the sign of \mbox{\(\langle V(\,\cdot-x_j)\psi_k,\psi_k\rangle\)}.

Fix \(\veps>0\). 
We define
\begin{equation}
\begin{split}
\tilde{\Omega}^{1}_{k,E}= &\left\{\omega\in\mathcal{C}_{\Lambda_n}\,\middle\vert\, c_j k(1+\veps\beta^{-}_{j,k})\leq m(\omega,B(x_j,k^{-\alpha}))\leq c_j k(1+\veps\beta^{+}_{j,k})\ (j=1,\ldots,q)\right\}\\
&\quad\cap \bigg\{\omega\in\mathcal{C}_{\Lambda_n}\,\bigg\vert\, m\bigg(\omega,\Lambda_l \setminus\bigcup_{j=1}^{q}B(x_j,k^{-\alpha})\bigg)=0\bigg\}.
\end{split}
\end{equation} 
We put \(\tilde{\Omega}_{k,E}=\tilde{\Omega}^{1}_{k,E}\cap\Omega^{2}_{E}\) (see \eqref{eq:Omega2} for the definition of \(\Omega^2_E\)).

Fix \(\omega\in\tilde{\Omega}_{k,E}\). We put
\begin{equation}
V^{(i)}_{\omega}=\sum_{y\in \omega\cap \Lambda_l}\tau_y V.
\end{equation}
We set \(m_j(\omega)=m(\omega,B(x_j,k^{-\alpha}))\) and \(\omega\cap B(x_j,k^{-\alpha})=\{y_{j,1},\ldots, y_{j,m_j(\omega)}\}\) for every \(j=1,\ldots,q\).
For large \(E>0\), since \(\omega\in \tilde{\Omega}^{1}_{k,E}\), from the definitions of \(\beta^{+}_{j,k},\ \beta^{-}_{j,k}\), it holds that
\begin{equation}
\begin{split}
&\langle (-\Delta+V^{(i)}_{\omega})\psi_k,\psi_k \rangle\\
&= \bigg\langle \bigg(-\Delta+\sum_{j=1}^{q}m_j(\omega)\tau_{x_j} V\bigg)\psi_k, \psi_k \bigg\rangle+\sum_{j=1}^{q}\sum_{i=1}^{m_j(\omega)} \langle (\tau_{y_{j,i}}V-\tau_{x_j}V)\psi_k, \psi_k \rangle \\
& \leq (1+\veps) \bigg\langle \bigg(-\Delta+k\sum_{j=1}^{q}c_j \tau_{x_j}V\bigg)\psi_k, \psi_k \bigg\rangle\\
&\qquad\qquad\qquad\qquad+(1+2\veps)\sum_{j=1}^{q} c_j k\sup_{|y|\leq k^{-\alpha}}|\langle (\tau_{y+x_j}V-\tau_{x_j}V)\psi_k, \psi_k \rangle|.
\end{split}
\end{equation}
Hence, from \eqref{eq:ground state} and the definition of the asymptotic ground state \(\psi_g\), we obtain
\begin{equation}
\langle (-\Delta+V^{(i)}_{\omega})\psi_k,\psi_k \rangle\leq (1+\veps)E_{\bm{c}_{q},\bm{x}_{q}}(k)+o(|E_{\bm{c}_{q},\bm{x}_{q}}(k)|) \quad (k\to +\infty).
\end{equation}

Setting \(k=g_{\bm{c}_q, \bm{x}_q}(E)\), we obtain
\begin{equation}\label{eq:energy lower bound}
\langle (-\Delta+V^{(i)}_{\omega})\psi_k, \psi_k \rangle \leq -E-2
\end{equation}
for sufficiently large \(E>0\).

From \cite{KP} and the fact that \(P^{Poi(z)}_{\Lambda_n}\) stochastically dominates \(P^{per}_{\Lambda_n}\), for some \(C>0\), we have
\begin{equation}\label{eq:pairwise Omega2}
P^{per}_{\Lambda_n}(\mathcal{C}_{\Lambda_n}\setminus\Omega^2_E) \leq P^{Poi(z)}_{\Lambda_n}(\mathcal{C}_{\Lambda_n}\setminus\Omega^2_E)
\leq C\frac{(z r_0^d)^{E^{\rho}}}{\lfloor E^{\rho}\rfloor!},
\end{equation}
where the constant \(z>0\) is the activity defined in \eqref{eq:pairwise}.

For sufficiently large \(E>0\), since the balls \(B(x_j, k^{-\alpha})\) are pairwise disjoint and contained in \(\Lambda_{l/2}\), we have
\begin{equation}
\begin{split}
U_{\Lambda_l,\gamma}(\eta)&=U(\eta)\\
& \leq \frac{1}{2}\sum_{1\leq i,j \leq q}\left(\varlimsup_{x\to x_i-x_j}\varphi(x)+\veps\right)c_i c_j k^2(1+2\veps)^2 -k(1+2\veps)z',
\end{split}
\end{equation}
for any \(\eta \in \tilde{\Omega}^{1}_{k,E} \cap \mathcal{C}_{\Lambda_l}\) and \(\gamma \in \mathcal{C}_{\Lambda_n}\), where we set \(z'=\min\{\log z, 0\}\).

Consequently, noting that \(l< n-2R\) for large \(E\), we apply the DLR equation \eqref{eq:per DLR} and the estimate \(Z_{\Lambda_l,\gamma}\leq \exp(|\Lambda_l|(z-1))\) to obtain
\begin{equation}\label{eq:prob Omega1}
\begin{split}
&P^{per}_{\Lambda_n}(\tilde{\Omega}^{1}_{k,E})\\
&\geq \exp\left(-\frac{1}{2}\sum_{1\leq i,j\leq q}\left(\varlimsup_{x\to x_i-x_j}\varphi(x)+\veps\right)c_i c_j k^2 (1+2\veps)^2 +k(1+2\veps)z'\right) \\ 
&\qquad\times \frac{|B(0,k^{-\alpha})|^{k(1+2\veps)}}{(\lfloor k(1+3\veps)\rfloor !)^q}e^{-z|\Lambda_l|}.
\end{split}
\end{equation}
Since \(\rho>0\) is sufficiently large, by \eqref{eq:pairwise Omega2}, \eqref{eq:prob Omega1}, and the fact that for some \(C>0\),
\begin{equation}
\frac{1}{C} E^{(2p-d)/2p} \leq g_{\bm{c}_q, \bm{x}_q}(E) \leq CE 
\end{equation}
for sufficiently large \(E>0\) (see the Appendix in \cite{KP}),
we get
\begin{equation}\label{eq:pairwise lower prob}
\begin{split}
&\log P^{per}_{\Lambda_n}(\tilde{\Omega}_{k,E})\\
&\geq \log \left( P^{per}_{\Lambda_n} (\tilde{\Omega}_{k,E}^1) - P^{per}_{\Lambda_n}(\mathcal{C}_{\Lambda_n}\setminus \Omega_E^2) \right)\\
&\geq -\frac{(1+2\veps)^2}{2}\sum_{1\leq i,j \leq q}c_i c_j \left(\varlimsup_{x\to x_i-x_j}\varphi(x)+\veps\right) g_{\bm{c}_q, \bm{x}_q}(E)^2 (1+o(1))\quad (E\to+\infty).
\end{split}
\end{equation} 

Combining \eqref{eq:energy lower bound}, \eqref{eq:pairwise lower prob}, and the same argument as in the proof of Proposition 3.1 in \cite{KP}, and noting that \(\veps>0\) is arbitrary, we obtain
\begin{equation}
\varliminf _{E\to+\infty} \frac{\log \int_{\mathcal{C}_{\Lambda_n}} N_{\omega,n}(-E-1) P^{per}_{\Lambda_n}(d\omega)}{g_{\bm{c}_q,\bm{x}_q}(E)^{2}}
\geq  -\frac{1}{2} \sum_{1\leq i,j\leq q} c_i c_j \varlimsup_{x\to x_i-x_j}\varphi(x). 
\end{equation}
Theorem \ref{thm:per app} then yields \eqref{eq:lower bound}, following the derivation of Proposition 3.2 from Proposition 3.1 in \cite{KP}.
\end{proof}

To prove Theorem \ref{thm:pairwise} \textbf{(b)}, we introduce the following notation and prepare two lemmas.
Throughout the remainder of this section, to simplify the notation, we use \(\langle\cdot,\cdot\rangle\) and \(\|\cdot\|\) to denote the standard inner product and norm on \(L^2(\Lambda_n)\), respectively.
Let \(T_n^{*}\) denote the subset \([0,2\pi/n)^d\subset\mathbb{R}^d\). For each \(n>0\), \(\omega\in\mathcal{C}_{\Lambda_n}\), and \(\theta\in T_n^{*}\), we introduce the self-adjoint operator \(H_{\omega,n,\theta}\) uniquely defined by the quadratic form \(\|\nabla \phi\|^2  + \langle V_{\omega,n}\phi, \phi \rangle\) on the space \(L^2_{\theta}(\Lambda_n)\), given by
\begin{equation}
L^2_{\theta}(\Lambda_n)=\{\phi\in L^{2}_{\mathrm{loc}}(\mathbb{R}^d) \mid \text{\(\phi(x+n\gamma)= e^{in\theta\cdot\gamma}\phi(x)\) for any \(x\in\mathbb{R}^d\) and \(\gamma\in \mathbb{Z}^d\)}\}.
\end{equation}
Furthermore, we set
\begin{equation}
\Xi_{l}= \{y_1-y_2 \mid y_1\in \Lambda_l,\ y_2\in\supp V_{2,-}\},
\end{equation}
and \(\Xi_l (x)=\tau_x \Xi_{l}\).

\begin{lemma}\label{lem:operator}
Suppose that the assumptions of Theorem \ref{thm:pairwise} hold.
There exist \(n_0, k_0>0\) such that if \(n>n_0\), \(k>k_0\), and \(0<l<1\), then for any \(\theta\in T_n^{*}\) and any \(\omega\in\mathcal{C}_{\Lambda_n}\) satisfying \(m(\pi_n(\omega),\Xi_{2l}(x))\leq k\) for all \(x\in\Lambda_n\), it holds that
\begin{equation}\label{eq:operator bound}
H_{\omega,n,\theta}\geq -E_{-}(k)-C\left(\frac{1}{l^2}+\frac{k}{l^{d}}\right),
\end{equation}
where \(C>0\) is a constant depending only on \(d\) and \(V\). 
\end{lemma}

\begin{proof}[Proof of Lemma \ref{lem:operator}]
Fix \(\theta\in T_n^{*},\ \phi\in C^{\infty}(\mathbb{R}^d)\cap L^2_{\theta}(\Lambda_n)\) with \(\|\phi\|=1\), and \(\omega\in\mathcal{C}_{\Lambda_n}\) such that \(m(\pi_n(\omega),\Xi_{2l}(x))\leq k\) for all \(x\in\Lambda_n\).
For simplicity of notation, we assume that \(n/l\) is an even integer; otherwise, one can slightly adjust the value of \(l\) without affecting the order of the subsequent estimates.

Following an argument similar to that in the proof of Lemma 4.2 in \cite{KP}, we proceed to estimate \(\langle H_{\omega,n,\theta}\phi,\phi\rangle\) via the IMS localization technique.
We introduce a partition of unity on \(\mathbb{R}^d\) given by
\begin{equation}
1=\sum_{\bm{j}\in \mathbb{Z}^d} \chi (x-\bm{j})^2,
\end{equation}
where \(\chi\in C^{\infty}_{0}(\mathbb{R}^d)\), \(0\leq \chi \leq 1\), \(\chi=1\) on \(\Lambda_{1/2}\) and \(\supp \chi\subset \Lambda_{3/2}\).
For a given \(l>0\), we define \(\chi_{\bm{j},l}(x)=\chi(x/l-\bm{j})\), yielding
\begin{equation}
1= \sum_{\bm{j}\in \mathbb{Z}^d}\chi_{\bm{j},l}^2.
\end{equation}
From \cite{KP}, we have
\begin{equation}
\langle H_{\omega,n,\theta}\phi,\phi \rangle \geq \sum_{\bm{j}\in\mathbb{Z}^d}\left(\|\nabla(\chi_{\bm{j},l}\phi)\|^2
+\langle V_{\omega,n}\chi_{\bm{j},l}\phi,\chi_{\bm{j},l}\phi \rangle\right) - \frac{C'}{l^2},
\end{equation}
where \(C'>0\) is a constant depending only on \(\chi\).

For each \(\bm{j}\in\mathbb{Z}^d\), we define
\begin{equation}
V^{(i,\bm{j})}_{\omega,n}= \sum_{y\in\pi_n(\omega)\cap \Xi_{2l}(l\bm{j}) }\tau_y V,\qquad
V^{(e,\bm{j})}_{\omega,n}=\sum_{y\in\pi_n(\omega)\cap (\Xi_{2l}(l\bm{j}))^c}\tau_y V.
\end{equation}
As \(\supp \chi_{\bm{j},l}\subset \Lambda_{3l/2}(l\bm{j})\), we have
\begin{equation}
\langle V^{(e,\bm{j})}_{\omega,n}\chi_{\bm{j},l}\phi, \chi_{\bm{j},l}\phi \rangle \geq -C''kl^{-d}\|\chi_{\bm{j},l}\phi\|^2,
\end{equation}
where \(C''>0\) is a constant independent of \(\phi\), \(n\), \(l\), \(k\), and \(\bm{j}\).
This lower bound follows from Condition \hyperref[cd:V1]{\textbf{(V1)}} and \(\max_{x\in\mathbb{R}^d} m(\pi_{n}(\omega), \Lambda_{2l}(x))\leq k\).

From \cite{KP}, it follows that
\begin{equation}
\sum_{\bm{j}\in\mathbb{Z}^d } \left( \|\nabla(\chi_{\bm{j},l}\phi)\|^2+ \langle V^{(i,\bm{j})}_{\omega,n}\chi_{\bm{j},l}\phi, \chi_{\bm{j},l}\phi \rangle \right)
\geq -E_{-}(k).
\end{equation}
Since \(V_{\omega,n}=V^{(i,\bm{j})}_{\omega,n}+V^{(e,\bm{j})}_{\omega,n}\) on \(\supp \chi_{\bm{j}, l}\), for sufficiently large \(k>0\), we obtain
\begin{equation}\label{eq:theta bound}
\langle H_{\omega,n,\theta}\phi,\phi\rangle 
\geq -E_{-}(k) -C\left(\frac{1}{l^2}+\frac{k}{l^{d}}\right),
\end{equation}
for some \(C>0\) depending only on \(d\) and \(V\).
By the proof of Lemma 4.2 in \cite{KP}, \eqref{eq:theta bound} implies \eqref{eq:operator bound}.
\end{proof}

\begin{lemma}\label{lem:prob}
Suppose that the assumptions of Theorem \ref{thm:pairwise} \textbf{(b)} hold.
For any \(0<\veps<1\), \(a\in A\), and \(\delta>0\), there exists \(n_0>0\) such that for all \(n\geq n_0\), \(k\geq n^{\delta}\) and \(k^{-1/\delta}<l<1\), it holds that
\begin{equation}
\log P^{per}_{\Lambda_n} (\Omega_{n})\leq -\frac{(1-\veps)a}{2X_{2l,a}}k^2,
\end{equation}
where \(\Omega_{n}\) is the event defined by
\begin{equation}
\Omega_{n}= \{\omega\in\mathcal{C}_{\Lambda_n} \mid \text{there exists \(x\in\Lambda_n\) such that \(m(\pi_n(\omega),\Xi_l(x))>k\)} \}.
\end{equation}
\end{lemma}

\begin{proof}[Proof of Lemma \ref{lem:prob}]
By \eqref{eq:per GPP stationary}, we have
\begin{equation}\label{eq:prob split}
\begin{split}
&P^{per}_{\Lambda_n}(\Omega_{n})\\
&\leq P^{per}_{\Lambda_n}\big(\{\omega\in\mathcal{C}_{\Lambda_n}\mid \text{there exists \(\bm{j}\in l\mathbb{Z}^d \cap \Lambda_{n+2l}\) such that \(m(\pi_n(\omega),\Xi_{2l}(\bm{j}))\geq k\)}\}\big)\\
& \leq (2+ n/l )^d P^{per}_{\Lambda_n}(K),
\end{split}
\end{equation}
where \(K=\{\omega\in\mathcal{C}_{\Lambda_n}\mid m(\pi_n(\omega),\Xi_{2l})\geq k\}\). 

Fix \(a\in A\). For sufficiently large \(n,k>0\), using the DLR equation \eqref{eq:per DLR}, we have
\begin{equation}\label{eq:prob K}
\begin{split}
P^{per}_{\Lambda_n}(K)
&= \int_{\mathcal{C}_{\Lambda_n}} \frac{1}{Z_{\Xi_{2l},\gamma}}\int_{\mathcal{C}_{\Xi_{2l}}}\mathbb{1}_K \big(\eta\cup\gamma_{\Xi_{2l}^c}\big)\exp\big(-U_{\Xi_{2l} ,\gamma}(\eta)\big) P^{Poi}_{\Xi_{2l}} (d\eta) P^{per}_{\Lambda_n}(d\gamma)\\
&\leq e^{z|\Xi_{2l}|}  \exp\Bigg(-\inf_{\substack{\zeta \subset \Xi_{2l} \\ \#\zeta=\lfloor k\rfloor}} \frac{1}{2}\sum_{\substack{x,y\in\zeta\\ x\neq y}}\varphi(x-y)\Bigg) \\
& \leq e^{z|\Xi_{2l}|} \exp\Bigg(-\inf_{\substack{\zeta\subset\Xi_{2l} \\ \# \zeta=\lfloor k\rfloor}}  \frac{a}{2}\sum_{\substack{x,y\in\zeta\\ x\neq y}} \mathbb{1}_{S_a}(x-y)\Bigg),
\end{split}
\end{equation}
where we used \(0\leq a\mathbb{1}_{S_a} \leq \varphi\) and \(Z_{\Xi_{2l},\gamma}\geq e^{-|\Xi_{2l}|}\) (with \(Z_{\Xi_{2l},\gamma}\) defined as in \eqref{eq:partition function}).

We next derive a lower bound for the infimum above. Let \(\zeta \in \mathcal{C}_{\Xi_{2l}}\) be a configuration with exactly \(\lfloor k \rfloor\) points. We view \(\zeta\) as the vertex set of a simple graph \(G=(\zeta, E)\), where an edge exists between distinct \(x, y \in \zeta\) if and only if \(x-y \in S_a\).
Let \(\overline{G}=(\zeta,\overline{E})\) be the complement graph of \(G\), where an edge \(\{x,y\} \in \overline{E}\) exists if and only if \(x-y \notin S_a\). 

By the definition of \(X_{2l,a}\), any independent set of \(G\) (i.e., a subset of \(\zeta\) with no edges between its elements) has size at most \(X_{2l,a}\).
This implies that \(\overline{G}\) cannot contain a complete graph of size \(X_{2l,a}+1\).

Recall that Tur\'{a}n's theorem (see, e.g., \cite{Diestel}) provides an upper bound on the number of edges for any graph that does not contain a complete subgraph of a specified size. Applying this to our complement graph \(\overline{G}\), which has \(\lfloor k \rfloor\) vertices and lacks a complete subgraph of size \(X_{2l,a}+1\), we find that the number of edges \(|\overline{E}|\) is bounded by
\begin{equation}
|\overline{E}| \leq \frac{1}{2}\left(1 - \frac{1}{X_{2l,a}}\right) \lfloor k \rfloor^2.
\end{equation}
Since the number of edges in \(G\) is given by \(|E| = \lfloor k \rfloor (\lfloor k \rfloor - 1)/2 - |\overline{E}|\), we obtain
\begin{equation}\label{eq:energy estimate}
\inf_{\substack{\zeta\subset \Xi_{2l} \\ \#\zeta=\lfloor k\rfloor}}\frac{1}{2} \sum_{\substack{x,y\in\zeta \\ x\neq y}} \mathbb{1}_{S_a}(x-y) = \inf_{\zeta} |E| \geq \frac{\lfloor k\rfloor(\lfloor k\rfloor-X_{2l,a})}{2X_{2l,a}}.
\end{equation}

Fix \(0<\veps<1\) and \(\delta>0\). Since \(l > k^{-1/\delta}\) and \(k \geq n^{\delta}\), we find that \((2+n/l)^d\) is of polynomial order in \(k\). 
By \eqref{eq:prob split}, \eqref{eq:prob K}, and \eqref{eq:energy estimate}, for sufficiently large \(n>0\), it holds that
\begin{equation}
\log P^{per}_{\Lambda_n}(\Omega_{n})\leq -\frac{(1-\veps)a}{2X_{2l,a}}k^2 .
\end{equation}
This completes the proof.
\end{proof}

\begin{proof}[Proof of Theorem \ref{thm:pairwise} \textbf{(b)}]
Fix \(0<\veps<1\) and \(a\in A\). For any sufficiently large \(E>0\), we set
\begin{equation} 
k=g((1-\veps)(E-2)),\quad
l=\left(\frac{k}{E-2}\right)^{1/2d}.
\end{equation}
Combining \cite[Lemma 5.5]{KP} with a simple calculation and noting that \(\essinf V= -\infty\), we have, for some constant \(C>0\), \(E^{-1/4p}/C\leq l = o(1)\) and \(E^{(2p-d)/2p}/C \leq g(E)=o(E)\) as \(E\to+\infty\), where \(p\) is given in Condition \hyperref[cd:V2]{\textbf{(V2)}}.
We set \(n = \lfloor E^{\beta_0} l \rfloor\), where \(\beta_0>0\) is a constant. For a given \(\beta\) in Theorem \ref{thm:per app}, by choosing \(\beta_0>0\) sufficiently large, we can ensure that \(n\geq E^{\beta}\) for sufficiently large \(E>0\).

From Lemma \ref{lem:operator}, for large \(E>0\), we get \(H_{\omega,n,\theta}\geq -E+2\) 
for all \(\theta\in T_n^{*}\) and all \(\omega\in\Omega^2_E\) such that \(m(\pi_n(\omega), \Xi_{2l} (x))\leq k\) for all \(x\in\Lambda_n\) (the definition of \(\Omega^2_E\) is \eqref{eq:Omega2}). 

By an argument similar to that in the proof of Proposition 4.1 in \cite{KP}, it follows that
\begin{equation}
\begin{split}
&\int_{\Omega^2_E} N_{\omega,n}(-E+1) P^{per}_{\Lambda_n}(d\omega)\\
& \leq E^{\alpha} P^{per}_{\Lambda_n}\big(\{\omega\in\mathcal{C}_{\Lambda_n}\mid \text{there exists \(x\in\Lambda_n\) such that \(m(\pi_n(\omega), \Xi_{2l} (x))>k\)}\}\big),
\end{split}
\end{equation}
for some \(\alpha>0\) independent of \(E\).
Pick a sufficiently large constant \(\rho>0\).
By an argument similar to that in the proof of Lemma 4.1 in \cite{KP} and the fact that \(P^{per}_{\Lambda_n}\) is stochastically dominated by \(P^{Poi}_{\Lambda_n}\), there exists a constant \(C>0\) such that for sufficiently large \(E>0\), it follows that
\begin{equation}
\int_{\mathcal{C}_{\Lambda_n}} N_{\omega,n}(-E+1)P^{per}_{\Lambda_n}(d\omega)
\leq \int_{\Omega^2_E} N_{\omega,n}(-E+1)P^{per}_{\Lambda_n}(d\omega) + \frac{C}{\lfloor E^{\rho}/2\rfloor !}.
\end{equation}
By Lemma \ref{lem:prob}, we obtain
\begin{equation}
\log\int_{\mathcal{C}_{\Lambda_n}} N_{\omega,n}(-E+1)P^{per}_{\Lambda_n}(d\omega) \leq -\frac{(1-\veps) a}{2X_{4l,a}}  g ((1-\veps)(E-2))^2 (1+o(1)) \quad (E\to+\infty).
\end{equation}
By Theorem \ref{thm:per app} (taking \(\nu>0\) sufficiently large), we have
\begin{equation}
\log N(-E) \leq -\frac{(1-\veps)^2 a}{2X_{4l,a}} g ((1-\veps)(E-2))^2 (1+o(1))\quad (E\to+\infty).
\end{equation}

From the Appendix in \cite{KP}, we have \(g(E-2) \leq g(E) \leq g(E-2) + C\) for some \(C > 0\), which implies \(g(E-2)\sim g(E)\) as \(E\to +\infty\). 
Since \(g\) is concave on \([E_0,+\infty)\) for some \(E_0>0\) from the variational principle, we have
\begin{equation*}
\frac{g(E)-g(E_0)}{E-E_0} \leq \frac{g((1-\veps)E) - g(E_0)}{(1-\veps)E-E_0},
\end{equation*}
which implies that 
\((1-\veps)g(E) \leq g((1-\veps)E)(1+o(1))\) as \(E\to+\infty\).
Thus, we obtain \((1-\veps) g(E) \leq g((1-\veps)(E-2)) (1+o(1))\) as \(E\to+\infty\).

Taking \(E \to +\infty\) (which implies \(l \downarrow 0\), and hence \(X_{4l,a} \to T_a\)), and subsequently letting \(\veps \downarrow 0\), we obtain
\begin{equation}
\varlimsup_{E \to +\infty} \frac{\log N(-E)}{g(E)^2} \leq -\frac{a}{2T_a}.
\end{equation}
Since this inequality holds for any \(a \in A\), we obtain the upper bound \eqref{eq:upper bound}.
\end{proof}

\begin{proof}[Proof of Proposition \ref{prop:pairwise} \textbf{(a)}]
Fix \(a\in(0,a_0)\). Note that \(a\in A\) by definition.
Since \(\Int S_{a}\) contains the compact set \(\{x-y\mid x,y\in \supp V_{2,-}\}\),
it follows that for sufficiently small \(l>0\),
\begin{equation}
\{x-y \mid x,y \in \Lambda_l + \supp V_{2,-}\} \subset S_{a}.
\end{equation}
This implies \(X_{l,a}=1\) and hence \(T_{a}=1\). Consequently, we obtain
\begin{equation}\label{eq:sup=1}
\sup_{a\in A} \frac{a}{T_a} =a_0.
\end{equation} 
Combining this with the upper bound in Theorem \ref{thm:pairwise} \textbf{(b)} and the corresponding lower bound in Theorem \ref{thm:pairwise} \textbf{(a)}, we obtain \eqref{eq:IDS pairwise}.
\end{proof}

\begin{proof}[Proof of Proposition \ref{prop:pairwise} \textbf{(b)}]
Fix \(a\in (0,a_0)\). Note that \(a\in A\). Since \(\lim_{x\to 0}\varphi(x) = a_0 > a\), we can choose \(r>0\) such that \(\Lambda_{4r}\subset S_{a}\). Since \(V\) is bounded outside \(\Lambda_{r}(x_0)\), we can choose a decomposition of \(V\) such that \(\supp V_{2,-} \subset \Lambda_r (x_0)\).
For any \(l\in(0,r)\), the set \(\Lambda_l + \supp V_{2,-}\) is contained in \(\Lambda_{2r}(x_0)\). Thus, for any \(x, y \in \Lambda_l + \supp V_{2,-}\), we have \(x-y \in \Lambda_{4r} \subset S_a\). This implies that \(X_{l,a}=1\) for all sufficiently small \(l>0\), which yields \(T_{a}=1\). 
Combining this with Theorem \ref{thm:pairwise} \textbf{(a)} and \textbf{(b)}, we obtain \eqref{eq:IDS pairwise}.
\end{proof}

\section{Improvement on the upper bound}\label{sec:improve}
In this section, we improve Theorem \ref{thm:pairwise} \textbf{(b)} under additional assumptions on the single-site potential \(V\), resulting in Theorem \ref{thm:V5}.

\subsection{Main results}
We consider \(V\) satisfying Condition \hyperref[cd:V1]{\textbf{(V1)}} and the following condition: 
\begin{description}
\item[\phantomsection \textbf{(V5)}\label{cd:V5}]
\(V_2\) is given by 
\begin{equation}
V_2= \sum_{i=1}^{q} b_i V_3^{(i)} +V_4 +V_5
\end{equation}
where \(q\in\mathbb{Z}_{>0}\) and \(b_i >0\) for \(i=1,\ldots,q\), and we assume the following:
\begin{itemize}
\setlength{\itemsep}{0.5em}
\item \(V_3, V_4, V_5\), and \(V_3^{(1)},\ldots, V_3^{(q)}\) are compactly supported functions in \(L^{p}(\mathbb{R}^d)\) with \(p>p(d)\), where \(p(d)\) is the constant given in \eqref{eq:p(d)}. Furthermore, the supports \(\supp V_3^{(1)},\ldots, \supp V_3^{(q)},\supp V_5\) are pairwise disjoint;
\item for each \(i=1,\ldots,q\), it holds that \(\essinf  V_3=\essinf  V_3^{(i)}=-\infty\). In contrast, \(V_4\) satisfies \(\essinf  V_4<0<\esssup V_4\), and \(V_5\) satisfies \(\essinf V_5 \geq 0\);
\item as \(E\to +\infty\), \(g_{3}^{(i)}(E)\sim g_3(E)\) for each \(i=1,\ldots,q\), and \(g_3(E)=o(g_4(E))\).
\end{itemize}
\end{description}
Here, \(g_3^{(i)}(E)\), \(g_3 (E)\) and \(g_4 (E)\) are the inverse functions of \(E_{-,3}^{(i)}(g)=-E_{3}^{(i)}(g)\), \(E_{-,3}(g) = -E_{3}(g)\), and \(E_{-,4}(g) = -E_{4}(g)\), respectively (cf. the definition of \(g(E)\)), where \(E_{3}^{(i)}(g)\), \(E_3 (g)\), and \(E_4 (g)\) denote the minimum eigenvalues of \(-\Delta+gV_3^{(i)}\), \(-\Delta+gV_3\), and \(-\Delta+gV_4\), respectively.

Under Condition \hyperref[cd:V5]{\textbf{(V5)}}, we define the following condition:
\begin{description}
\item[\phantomsection \textbf{(V6)}\label{cd:V6}]
\begin{itemize}
\setlength{\itemsep}{0.5em}
\item there exist distinct points \(y_1,\ldots, y_q\) such that \(V^{(i)}_3=\tau_{y_i}V_3\) for each \(i=1,\ldots,q\);
\item \(g_3(E)=o(\tilde{g}_4(E))\quad (E\to +\infty)\).
\end{itemize}
\end{description}
Here, \(\tilde{g}_4(E)\) is the inverse function of \(\tilde{E}_{-,4}(g)(=-\tilde{E}_{4}(g))\), and \(\tilde{E}_{4}(g)\) denotes the infimum of the spectrum of \(-\Delta - gV_4\).
We note that \(\tilde{g}_{4}(E)\) is well-defined for sufficiently large \(E\) because \(\essinf  (-V_{4})<0\) holds under Condition \hyperref[cd:V5]{\textbf{(V5)}}. 

We define the set \(\tilde{A}\) by 
\begin{equation}
\tilde{A}=\{a>0\mid x-y\in\Int S_a \text{ whenever } x,y\in \supp V_3^{(i)} \text{ for each } i=1,\ldots, q\},
\end{equation}
where \(S_a\) is defined in \eqref{eq:S_a}.
We note that \(\tilde{A}\) may be empty, and that \(\tilde{A}\neq\emptyset\) implies that the supports of \(V_3^{(i)}\) are small enough compared to the interaction range.

We put
\begin{equation}
I_a= \left\{(i,j)\in \{1,\ldots, q\}^2 \,\middle\vert\, \text{\(x-y\in \Int  S_a\) for all \(x\in \supp V_3^{(i)}\) and \(y\in \supp V_3^{(j)}\)}\right\}.
\end{equation}
For any \(I'\subset\{1,\ldots,q\}^2\), we define
\begin{equation}
K(I')=\{J\subset \{1,\ldots, q\}\mid \text{\((i,j)\notin I'\) whenever \(i, j\in J\) and \(i\neq j\)}\}.
\end{equation}

Under Conditions \hyperref[cd:V5]{\textbf{(V5)}}, \hyperref[cd:V6]{\textbf{(V6)}}, and \hyperref[cd:L]{\textbf{(L)}}, we introduce the following condition on \(V\) and \(\varphi\):
\begin{description}
\item[\phantomsection \textbf{(V--P1)}\label{cd:V-P1}]
The following equalities hold:
\begin{itemize}
\setlength{\itemsep}{0.5em}
\item \(a_0=\sup \tilde{A}\);
\item \(\displaystyle \lim_{a\uparrow a_0} \max_{J\in K(I_a)} \sum_{i\in J} b_i^2 =\max_{J\in K(I)} \sum_{i\in J}b_i^2\),
\end{itemize}
where we set
\begin{equation}
I= \{(i,j)\in \{1,\ldots,q\}^2 \mid y_i-y_j\in \supp \varphi\}.
\end{equation}
\end{description}
We note that \(a\mapsto \max_{J\in K(I_a)} \sum_{i\in J} b_i^2\) is non-decreasing.

In \cite{KP}, the leading term of \(\log N(-E)\) under Conditions \hyperref[cd:V5]{\textbf{(V5)}} and \hyperref[cd:V6]{\textbf{(V6)}} remains undetermined for the Poisson case, except when Condition \hyperref[cd:V4]{\textbf{(V4)}} is additionally satisfied. In contrast, we determine it for certain Gibbs point processes with pairwise interactions under Condition \hyperref[cd:V-P1]{\textbf{(V--P1)}}.

\begin{theorem}\label{thm:V5}
Assume that \(\Gamma\) is a Gibbs point process with a pairwise energy function \(U\) defined in \eqref{eq:pairwise} satisfying Conditions \hyperref[cd:U1]{\textbf{(U1)}} and \hyperref[cd:U3]{\textbf{(U3)}}, and let \(V\) satisfy Conditions \hyperref[cd:V1]{\textbf{(V1)}} and \hyperref[cd:V5]{\textbf{(V5)}}.
Furthermore, suppose \(\tilde{A}\neq \emptyset\).
Then, the following results hold.
\begin{itemize}
\item[\textbf{(a)}] It holds that
\begin{equation}
\log N(-E)\leq -\frac{1}{2}\,\sup_{a\in \tilde{A}}\frac{a}{\displaystyle \max_{J\in K(I_a)}\sum_{i\in J}b_i^2}\, g_{3}(E)^2 (1+o(1))\quad (E\to+\infty).
\end{equation}
\item[\textbf{(b)}] Assume further that \(V\) satisfies Conditions \hyperref[cd:V6]{\textbf{(V6)}} and \hyperref[cd:L]{\textbf{(L)}}, and that Condition \hyperref[cd:V-P1]{\textbf{(V--P1)}} is satisfied.
Then, we have
\begin{equation}\label{eq:V6 IDS}
\log N(-E)\sim -\frac{1}{2}\frac{a_0}{\displaystyle \max_{J\in K(I)}\sum_{i\in J}b_i^2}\, g_{3}(E)^2 \quad (E\to +\infty).
\end{equation}
\end{itemize}
\end{theorem}

We now focus on the case where each \(V^{(i)}_3\) has a unique singular point at \(y_i\); that is, \(V^{(i)}_3\) is bounded outside any neighborhood of \(y_i\).

For \(a>0\), we define the set \(\tilde{I}_a\) by
\begin{equation}
\tilde{I}_a=\{(i,j)\in \{1,\ldots,q\}^2\mid y_i-y_j \in \Int  S_a\}.
\end{equation}
Furthermore, under Condition \hyperref[cd:L]{\textbf{(L)}}, we introduce the following condition, which is obtained by replacing \(I_a\) with \(\tilde{I}_a\) in Condition \hyperref[cd:V-P1]{\textbf{(V--P1)}}:

\begin{description}
\item[\phantomsection \textbf{(V--P2)}\label{cd:V-P2}]
it holds that 
\begin{equation}
\lim_{a\uparrow a_0} \max_{J\in K(\tilde{I}_a)} \sum_{i\in J} b_i^2 =\max_{J\in K(I)} \sum_{i\in J}b_i^2.
\end{equation}
\end{description}

\begin{proposition}\label{prop:V6}
Suppose that the Gibbs point process \(\Gamma\) with the pairwise energy function in \eqref{eq:pairwise} satisfies Conditions \hyperref[cd:U1]{\textbf{(U1)}} and \hyperref[cd:U3]{\textbf{(U3)}}, and that \(\varliminf_{x\to 0}\varphi(x)>0\) holds. We assume that \(V\) satisfies Conditions \hyperref[cd:V1]{\textbf{(V1)}} and \hyperref[cd:V5]{\textbf{(V5)}}, and that each \(V_3^{(i)}\) has a unique singular point at \(y_i\) for \(i=1,\ldots,q\). 
Then, the following results hold.
\begin{itemize}
\item[\textbf{(a)}] We have
\begin{equation}\label{eq:V5 IDS}
\log N(-E)\leq -\frac{1}{2}\, \sup_{a\in A} \frac{a}{\displaystyle\max_{J\in K(\tilde{I}_a)}\sum_{i\in J}b_i^2}\, g_3(E)^2 (1+o(1)) \quad (E\to+\infty).
\end{equation}  
\item[\textbf{(b)}] Furthermore, we assume that \(V\) satisfies \hyperref[cd:V6]{\textbf{(V6)}}, and that \(\varphi\) satisfies \hyperref[cd:L]{\textbf{(L)}}. If Condition \hyperref[cd:V-P2]{\textbf{(V--P2)}} is satisfied,
then we have
\begin{equation}\label{eq:prop V6 IDS}
\log N(-E)\sim -\frac{1}{2}\frac{a_0}{\displaystyle\max_{J\in K(I)}\sum_{i\in J}b_i^2}\, g_3(E)^2 \quad (E\to+\infty).
\end{equation}
\end{itemize}
\end{proposition}

As a simple illustration of Proposition \ref{prop:V6} \textbf{(b)}, we consider a Strauss process (see Example \ref{ex:Strauss}) under a single-site potential \(V\) with two potential wells. As Proposition \ref{prop:V6} \textbf{(b)} shows, when each well has a unique singular point, the problem reduces to examining the distances between these points.

\begin{example}[Two potential wells]
Suppose that \(q=2\) and \(V\) has two deep wells whose unique singular points are located at \(y_1\) and \(y_2\), with the corresponding depth weights \(b_1>0\) and \(b_2>0\). That is, \(V=b_1 \tau_{y_1} V_3 + b_2 \tau_{y_2} V_3\) (we can choose \(V_1+V_4+V_5=0\) such that Conditions \hyperref[cd:V5]{\textbf{(V5)}} and \hyperref[cd:V6]{\textbf{(V6)}} are satisfied). Without loss of generality, we assume \(b_1 \geq b_2\).
We consider a Strauss process whose pairwise energy function is given by \(\varphi(x)= a_0 \mathbb{1}_{B(0,R)}(x)\) for some \(0<a_0<+\infty\) and \(0<R<\mu_d^{1/d}\), where \(\mu_d\) is given by \eqref{eq:def mu_d}.
In this case, the support of \(\varphi\) is \(B(0,R)\). By the definition of \(I\), we trivially have \((1,1), (2,2)\in I\). 
Here, we exclude the boundary case \(|y_1-y_2|=R\). In this specific case, Condition \hyperref[cd:V-P2]{\textbf{(V--P2)}} fails.
The asymptotic behavior of the integrated density of states depends heavily on the distance between the two singular points:

\begin{itemize}
\item \textbf{Case 1 (\(|y_1-y_2|<R\)):} The two singular points are close to each other. Fix \(0<a<a_0\). Since \(y_1-y_2\in \Int S_a\), we have \((1,2), (2,1)\in I = \tilde{I}_a\). Thus, \(K(I)=K(\tilde{I}_a)= \{ \emptyset, \{1\}, \{2\} \}\). By \eqref{eq:prop V6 IDS}, we obtain
\begin{equation}
\log N(-E)\sim -\frac{1}{2}\frac{a_0}{b_1^2}\, g_3(E)^2 \quad (E\to+\infty).
\end{equation}
In this case, the event in which points concentrate in a single cluster gives the dominant contribution to the leading term.

\item \textbf{Case 2 (\(|y_1-y_2|> R\)):} The two singular points are separated by a distance strictly greater than the interaction range. Fix \(0<a<a_0\). Since \((1,2), (2,1)\notin I = \tilde{I}_a\), the set \(K(I)=K(\tilde{I}_a)\) includes the set \(\{1,2\}\). Therefore, we obtain
\begin{equation}
\log N(-E)\sim -\frac{1}{2}\frac{a_0}{b_1^2+b_2^2}\, g_3(E)^2 \quad (E\to+\infty).
\end{equation}
In this case, the dominant contribution to the leading term comes from the event where points form two distinct clusters separated by the relative displacement \(y_1-y_2\). Due to this specific separation, the \(b_1 \tau_{y_1} V_3\) components generated by one cluster and the \(b_2 \tau_{y_2} V_3\) components generated by the other overlap around the same point to cooperatively form a single well. Although this resulting well is no deeper than the one in Case 1 (with equality holding when \(b_1=b_2\)), the probability of forming such separated clusters is higher than that of forming a single cluster due to the repulsive interaction. Consequently, the IDS decays more slowly than in Case 1.
\end{itemize}
\end{example}

\subsection{Proof of Theorem \ref{thm:V5} \textbf{(a)}}

We pick \(a\in \tilde{A}\). 
We set
\begin{equation}
\begin{split}
\Xi^{(i)}_{l}&=\{x_1-x_2\mid x_1\in \Lambda_l,\ x_2\in \supp V^{(i)}_3\},\\
\Xi_{l}^{*}&=\{x_1-x_2\mid x_1\in \Lambda_l,\ x_2\in \supp V_4\},
\end{split}
\end{equation}
where the constant \(l>0\) (depending on \(a\)) is chosen to satisfy the following condition:
\begin{description}
\item[\phantomsection \textbf{(S)}\label{cd:S}]
\begin{itemize}
\setlength{\itemsep}{0.5em}
\item \(\Xi_{4l}^{(1)},\ldots, \Xi_{4l}^{(q)}\) are pairwise disjoint;
\item \(\{x-y\mid x\in\Xi_{4l}^{(i)},\ y \in \Xi_{4l}^{(j)}\}\subset \Int  S_a\) whenever \((i,j)\in I_a\).
\end{itemize}
\end{description}
For any \(x\in\mathbb{R}^d\), we define \(\Xi^{(i)}_{l}(x)=\tau_x \Xi^{(i)}_{l}\) and \(\Xi_l^{*} (x)=\tau_x \Xi_{l}^{*}\).

Let \(n=\lfloor E^{\beta}\rfloor\), where \(\beta>0\) is a sufficiently large constant.
For any \(E,k>0\), we put
\begin{equation}
\hat{\Omega}_{E,k} =\bigg\{\omega\in\mathcal{C}_{\Lambda_n}\, \bigg\vert\, \text{there exists \(x\in\Lambda_n\) such that \(\sum_{i=1}^q b_i m\big(\pi_n(\omega),\Xi^{(i)}_{2l} (x)\big)\geq k\)}\bigg\}.
\end{equation}

The proof of Theorem \ref{thm:V5} \textbf{(a)} is a consequence of the following two lemmas.
\begin{lemma}\label{lem:V5 H}
Suppose that the assumptions of Theorem \ref{thm:V5} \textbf{(a)} hold.
For any \(\veps>0\), there exist \(E_0, k_0>0\) such that for any \(E>E_0\), \(k>k_0\), \(\omega\notin \hat{\Omega}_{E,k}\), and \(\theta\in T_n^{*}\), it holds that
\begin{equation}\label{eq:V5 H}
H_{\omega,n,\theta}\geq -E_{-,3}((1+\veps)k).
\end{equation}
\end{lemma}

\begin{proof}
In this proof, to simplify the notation, we use \(\langle\cdot,\cdot\rangle\) and \(\|\cdot\|\) to denote the standard inner product and norm on \(L^2(\Lambda_n)\), respectively.

Consider sufficiently large \(E>0\) (i.e., \(n>0\) is sufficiently large). Without loss of generality, we may replace \(l\) with a slightly smaller value such that \(n/l \in 2\mathbb{Z}\). Since decreasing \(l\) only shrinks the sets \(\Xi_{4l}^{(i)}\), Condition \hyperref[cd:S]{\textbf{(S)}} remains satisfied, and this adjustment does not affect the subsequent estimates.

We pick \(\veps>0\) and \(\phi\in C^{\infty}(\mathbb{R}^{d})\cap L^2_{\theta}(\Lambda_n)\) such that \(\|\phi\|=1\).
We put 
\begin{equation}
V^{(3)}_{\bm{j},\omega,n}= \sum_{i=1}^{q} \sum_{y\in \pi_n(\omega)\cap\Xi^{(i)}_{2l}(l \bm{j})} b_i \tau_{y} V_3^{(i)},\quad
V^{(4)}_{\bm{j},\omega,n}= \sum_{y\in \pi_n(\omega)\cap \Xi_{2l}^{*}(l \bm{j})} \tau_{y} V_4.
\end{equation}

Recall that
\begin{equation}
1= \sum_{\bm{j}\in \mathbb{Z}^d} \chi^{2}_{\bm{j}, l},
\end{equation}
which was defined in the proof of Lemma \ref{lem:operator}.
Note that \(\supp \chi_{\bm{j}, l} \subset \Lambda_{3l/2}(l\bm{j})\).

Fix \(\omega\notin \hat{\Omega}_{E,k}\). Since \(\essinf V_5\geq 0\) from Condition \hyperref[cd:V5]{\textbf{(V5)}}, we have
\begin{equation}\label{eq:V5 H1}
\begin{split}
&\left\langle (-\Delta+V_{\omega,n})\phi,\phi \right\rangle\\
&=\sum_{\bm{j}\in\mathbb{Z}^d} \Big( \|\nabla(\phi\chi_{\bm{j},l})\|^2 -\|\phi \nabla\chi_{\bm{j},l}\|^2 + \langle V_{\omega,n} \chi_{\bm{j},l} \phi, \chi_{\bm{j},l} \phi \rangle \Big)\\
&\geq \frac{1}{1+\veps}\sum_{\bm{j}\in \mathbb{Z}^d} \left( \|\nabla(\chi_{\bm{j},l} \phi)\|^2 +(1+\veps) \big\langle V^{(3)}_{\bm{j},\omega,n}\chi_{\bm{j},l}\phi, \chi_{\bm{j},l}\phi \big\rangle  \right)\\
&\qquad +\frac{\veps}{1+\veps} \sum_{\bm{j}\in \mathbb{Z}^d} \Big( \|\nabla(\chi_{\bm{j},l} \phi)\|^2 +(1+1/\veps) \big\langle V^{(4)}_{\bm{j},\omega,n}\chi_{\bm{j},l}\phi, \chi_{\bm{j},l}\phi \big\rangle \Big) -C_l k,
\end{split}
\end{equation}
where the constant \(C_l\) depends on \(l\) but is independent of \(E\) and \(k\), and we used the exponential decay of \(V_1\) and the fact that \(m(\pi_n(\omega),\Xi_{2l}^{(1)} (x))\leq k/b_1\) for any \(\omega\notin \hat{\Omega}_{E,k}\) and \(x\in\mathbb{R}^d\).

Fix \(\bm{j}\in \mathbb{Z}^d\). For every \(i=1,\ldots, q\), we put \(m_i=m(\pi_n(\omega),\Xi_{2l}^{(i)}(l \bm{j}))\), and let \(\{x^{(i)}_s \mid s=1,\ldots, m_i\}\) denote the set \(\pi_n(\omega)\cap \Xi_{2l}^{(i)}(l \bm{j})\).
For sufficiently large \(k>0\), by the variational principle, we obtain
\begin{equation}\label{eq:V5 H2}
\begin{split}
&\|\nabla(\chi_{\bm{j},l} \phi)\|^2 +(1+\veps) \big\langle V^{(3)}_{\bm{j},\omega,n}\chi_{\bm{j},l}\phi, \chi_{\bm{j},l}\phi \big\rangle\\
& = \left( \sum_{i'=1}^{q} m_{i'} b_{i'} \right)^{-1} \\
&\qquad\qquad \times \sum_{i=1}^{q} \sum_{s=1}^{m_i}  b_i \Bigg( \|\nabla(\chi_{\bm{j},l}\phi)\|^2 + (1+\veps) \left(\sum_{i'=1}^{q}m_{i'} b_{i'}\right) \big\langle \big(\tau_{x_s^{(i)}}V_3^{(i)}\big) \chi_{\bm{j},l}\phi,\chi_{\bm{j},l}\phi \big\rangle \Bigg)\\
& \geq -\left(\sum_{i'=1}^{q} m_{i'} b_{i'} \right)^{\!-1} \sum_{i=1}^{q} m_i b_i E_{-,3}^{(i)}\left((1+\veps) \sum_{i'=1}^{q}m_{i'}b_{i'}\right) \|\chi_{\bm{j},l} \phi\|^2\\
& \geq - E_{-,3}\left((1+\veps)^2 k \right)\|\chi_{\bm{j},l} \phi\|^2 ,
\end{split}
\end{equation}
where the last inequality follows from the fact that \(\sum_i b_i m_i \leq k\) and that
\begin{equation}
E^{(i)}_{-,3}(g) \leq E_{-,3}\left((1+\veps)g\right)\quad (i=1,\ldots ,q)
\end{equation}
holds for sufficiently large \(g>0\).
By a calculation similar to the above, we obtain
\begin{equation}\label{eq:V5 H3}
\|\nabla(\chi_{\bm{j},l} \phi)\|^2 + (1+1/\veps) \big\langle V^{(4)}_{\bm{j},\omega,n}\chi_{\bm{j},l}\phi, \chi_{\bm{j},l}\phi \big\rangle \geq -E_{-,4}\big((1+1/\veps)C_l' k \big)\|\chi_{\bm{j},l} \phi\|^2,
\end{equation}
where \(C_l'>0\) is a constant depending on \(l\) but independent of \(\bm{j}\), \(E\), and \(k\).

By Condition \hyperref[cd:V5]{\textbf{(V5)}}, \eqref{eq:V5 H1}, \eqref{eq:V5 H2}, \eqref{eq:V5 H3}, and the fact that
\begin{equation}
\lim_{g\to+\infty}\frac{E_{-,3}(g)}{g}=+\infty,
\end{equation}
we have \(\langle(-\Delta+V_{\omega,n})\phi, \phi \rangle \geq -E_{-,3}((1+\veps)^3 k)\) for sufficiently large \(k>0\), which implies that \(H_{\omega,n,\theta}\geq -E_{-,3}((1+\veps)^3 k)\). Since \(\veps>0\) is arbitrary, we obtain \eqref{eq:V5 H}.
\end{proof}

\begin{lemma}\label{lem:V5 prob}
Suppose that the assumptions of Theorem \ref{thm:V5} \textbf{(a)} hold.
For any \(0<\veps<1\), \(0<\delta<1\), and \(a\in\tilde{A}\), there exists \(E_0>0\) such that for all \(E>E_0\) and all \(k>E^{\delta}\), it holds that
\begin{equation}\label{eq:tilde Omega}
\log P^{per}_{\Lambda_n}(\hat{\Omega}_{E,k})\leq -(1-\veps) \frac{a}{2\max_{J\in K(I_a)}\sum_{i\in J}b_i^2}\,k^2.
\end{equation} 
\end{lemma}

\begin{proof}
Fix \(0<\veps<1\), \(0<\delta<1\), and \(a\in\tilde{A}\). Choose \(0<\veps'<\veps\). Using stationarity \eqref{eq:per GPP stationary}, for sufficiently large \(E>0\), we have
\begin{equation}\label{eq:Omega split}
\begin{split}
P^{per}_{\Lambda_n}(\hat{\Omega}_{E,k})
&\leq P^{per}_{\Lambda_n}\Bigg( \bigcup_{\bm{j}\in l\mathbb{Z}^d} \bigg\{ \omega\in \mathcal{C}_{\Lambda_n} \;\bigg\vert\; \sum_{i=1}^{q}b_i m(\pi_n(\omega), \Xi^{(i)}_{4l}(\bm{j}))\geq k \bigg\}\Bigg)\\
&\leq (2+n/l)^{d} P^{per}_{\Lambda_n}\bigg(\bigg\{\omega\in\mathcal{C}_{\Lambda_n} \;\bigg\vert\; \sum_{i=1}^q b_i m(\omega, \Xi^{(i)}_{4l})\geq k \bigg\} \bigg).
\end{split}
\end{equation}

We put \(\Xi=\bigcup_{i=1}^{q}\Xi^{(i)}_{4l}\). By the definition of \(P^{per}_{\Lambda_n}\) and the DLR equation \eqref{eq:per DLR}, for any non-empty set \(J'\subset \{1,\ldots,q\}\), there exists a positive integer \(k_0\) such that for any integer \(k_i\geq k_0\ (i\in J')\), it holds that
\begin{equation}\label{eq:V5 Omega1}
\begin{split}
&P^{per}_{\Lambda_n}\left(m(\, \cdot\, ,\Xi_{4l}^{(i)})=k_i\,;\, i\in J'\right)\\
&= \int_{\mathcal{C}_{\Lambda_n}} \frac{1}{Z_{\Xi,\gamma}} \int_{\mathcal{C}_{\Xi}} \mathbb{1}_{\big\{m (\, \cdot \, ,\Xi_{4l}^{(i)} )=k_i \, ; \, i\in J' \big\}} (\eta) \exp(-U_{\Xi, \gamma}(\eta))P^{Poi}_{\Xi}(d\eta) P^{per}_{\Lambda_n}(d\gamma)\\
&\leq \exp \! \left(-\frac{a}{2}\sum_{(i,j)\in I_a(J')}(k_i-1) (k_j-1)\right) e^{|\Xi|} z^{\sum_{i\in J'}k_i} P^{Poi}_{\Xi} (m(\, \cdot \, , \Xi_{4l}^{(i)})=k_i\,;\, i\in J')\\
&= \exp \! \left(-\frac{a}{2}\sum_{(i,j)\in I_a(J')}(k_i-1) (k_j-1)\right)\prod_{i\in J'} \frac{|\Xi|^{k_i}}{k_i !} z^{\sum_{i\in J'}k_i}\\
&\leq \exp \! \left(-\frac{a(1-\veps')^{1/2}}{2}\sum_{(i,j)\in I_a (J')}k_i k_j\right),
\end{split}
\end{equation}
where we set 
\begin{equation}
I_a(J') = \left\{(i,j)\in J' \times J'\mid \{x-y\mid x\in \Xi_{4l}^{(i)},\ y\in \Xi_{4l}^{(j)}\}\subset \Int  S_a \right\},
\end{equation}
and we have used the bound \(Z_{\Xi,\gamma}\geq \exp(-|\Xi|)\).

By an argument similar to that in the proof of Proposition 4.6 in \cite{Nakagawa}, for sufficiently large \(t>0\), we obtain
\begin{equation}
\log \int_{\mathcal{C}_{\Lambda_n}}\exp\left(t\sum_{i=1}^{q}b_i m(\omega,\Xi_{4l}^{(i)})\right)P^{per}_{\Lambda_n}(d\omega)\leq \frac{t^2}{2a(1-\veps')}\max_{J\in K(I_a)}\sum_{i\in J}b_i^2.
\end{equation}
This implies that
\begin{equation}
\begin{split}
&\log P^{per}_{\Lambda_n}\left(\sum_{i=1}^q b_i m(\omega, \Xi^{(i)}_{4l})\geq k\right)\\
&\leq \log \int_{\mathcal{C}_{\Lambda_n}}\exp\left(t\left(\sum_{i=1}^{q}b_i m(\omega,\Xi^{(i)}_{4l})-k\right)\right)P^{per}_{\Lambda_n}(d\omega)\\
&\leq -kt+ \frac{t^2}{2a(1-\veps')}\max_{J\in K(I_a)}\sum_{i\in J}b_i^2.
\end{split}
\end{equation}
Substituting
\begin{equation}
t=\frac{k}{2}\left(\frac{1}{2a(1-\veps')}\max_{J\in K(I_a)}\sum_{i\in J}b_i^2\right)^{-1}
\end{equation}
into the above inequality, we obtain for sufficiently large \(k>0\),
\begin{equation}\label{eq:V5 Omega2}
\log P^{per}_{\Lambda_n}\left(\sum_{i=1}^q b_i m(\omega, \Xi^{(i)}_{4l})\geq k\right)
\leq -(1-\veps')\frac{a}{2\max_{J\in K(I_a)} \sum_{i\in J}b_i^2}\, k^2.
\end{equation}
Since \(\log(2+n/l)^d=o(k^2)\) and \(\veps' < \veps\), taking \(E>0\) sufficiently large, we obtain \eqref{eq:tilde Omega} from \eqref{eq:Omega split} and \eqref{eq:V5 Omega2}.
\end{proof}

\begin{proof}[Proof of Theorem \ref{thm:V5} \textbf{(a)}]
Fix \(\veps>0\). We set
\begin{equation}
k=\frac{1}{1+\veps}\,g_3(E-2).
\end{equation}
We note that there exists \(0<\delta<1\) such that \(k>E^{\delta}\) for sufficiently large \(E>0\) (see \cite{KP}). Then, applying Lemmas \ref{lem:V5 H} and \ref{lem:V5 prob} along with Theorem \ref{thm:per app}, and using an argument similar to that in the last step of the proof of Theorem \ref{thm:pairwise} \textbf{(b)}, we complete the proof of Theorem \ref{thm:V5} \textbf{(a)}.
\end{proof}

\subsection{Proof of Theorem \ref{thm:V5} \textbf{(b)}}

We use the following two lemmas to prove the lower bound.

\begin{lemma}\label{lem:decay}
Assume that \(V_3\) is as in Condition \hyperref[cd:V5]{\textbf{(V5)}}.
Let \(\phi_g \in H^1(\mathbb{R}^d)\) be the normalized ground state of \(-\Delta+gV_3\). For any open set \(D\) containing \(\supp V_3\), there exist constants \(C > 0\) and \(g_0 > 0\), such that for all \(g > g_0\) and \(x \in \mathbb{R}^d \setminus D\),
\begin{equation}\label{eq:decay}
|\phi_g(x)| \leq C \exp\left( - \frac{g^{1/2}}{C} \right).
\end{equation}
\end{lemma}

\begin{proof}
Since \(V_3\) satisfies the Kato-class condition, \(\phi_g \in L^{\infty}(\mathbb{R}^d)\) (see, e.g., \cite{CL}). Using the standard integral representation via the Green's function for the equation \((-\Delta - E_3(g))\phi_g = -gV_3\phi_g\), we obtain the integral equation
\begin{equation}
\phi_g(x) = -g \int_{\mathbb{R}^d} G_{E_3(g)}(x, y) V_3 (y) \phi_g(y) dy,
\end{equation}
where \(G_{E_3(g)}\) is the free resolvent kernel. There exists \(\delta > 0\) such that for all \(x \in \mathbb{R}^d \setminus D\) and \(y \in \supp V_3\), \(|x-y| \geq \delta\). Away from the singularity at \(x=y\), the free resolvent kernel satisfies the exponential bound \(|G_{E_3(g)}(x, y)| \leq C_1 \exp(- |E_3(g)|^{1/2} |x-y|/C_1)\) for some constant \(C_1>0\) (see, e.g., \cite{CL, RS4}). Therefore, for sufficiently large \(g>0\), we can estimate
\begin{equation}
|\phi_g(x)| \leq g C_1 \exp\left(- \frac{\delta g^{1/2}}{C_1}\right) \int_{\mathbb{R}^d} |V_3(y)| |\phi_g(y)| dy,
\end{equation}
for \(x\in \mathbb{R}^d \setminus D\), where we used \(|E_3(g)|\geq g\) for sufficiently large \(g>0\) (see the Appendix in \cite{KP}).
By H\"{o}lder's inequality, the integral is bounded by \(\|V_3\|_{L^p(S)} \|\phi_g\|_{L^q(S)}\) with \(1/p + 1/q = 1\), where \(S=\supp V_3\). Since \(S\) is compact and \(p>2\) (i.e., \(1<q<2\)), applying H\"{o}lder's inequality again yields \(\|\phi_g\|_{L^q(S)}\leq |S|^{1/q - 1/2}\), where we used \(\|\phi_g\|=1\). 
Consequently, for sufficiently large \(g\), we obtain \eqref{eq:decay}.
\end{proof}

Before stating the lemma, we recall that \(D_{q'}\) and \(g_{\bm{c}_{q'},\bm{x}_{q'}}\) are defined in Section \ref{sec:pairwise}.

\begin{lemma}\label{lem:V6}
Suppose that \(V\) satisfies Conditions \hyperref[cd:V1]{\textbf{(V1)}}, \hyperref[cd:V5]{\textbf{(V5)}}, and \hyperref[cd:V6]{\textbf{(V6)}}. Let \(S_q\) denote the symmetric group of degree \(q\). Fix \(\sigma\in S_q\).
For any \(1\leq q' \leq q\) and \(c_1,\ldots, c_{q'}>0\) such that \(\sum_{j=1}^{q'}c_j=1\), it holds that \((\bm{c}_{q'}, \bm{x}_{q'})\in D_{q'}\) and 
\begin{equation}\label{eq:g estimate}
\varliminf_{E\to+\infty}\frac{g_3(E)}{g_{\bm{c}_{q'},\bm{x}_{q'}}(E)} \geq \sum_{j=1}^{q'}c_j b_{\sigma(j)},
\end{equation}
where \(\bm{c}_{q'}=(c_1,\ldots, c_{q'})\) and \(\bm{x}_{q'}=(x_1,\ldots, x_{q'})=(-y_{\sigma(1)},\ldots, -y_{\sigma(q')})\). 
\end{lemma}

\begin{proof}
Throughout the proof, \(C>0\) denotes a generic sufficiently large constant that may change from line to line.

We can estimate \(\sum_{j=1}^{q'} c_j \tau_{x_j}V\) as
\begin{equation}
\sum_{j=1}^{q'} c_j \tau_{x_j}V
\leq \sum_{j=1}^{q'} c_j b_{\sigma(j)} V_3 +\sum_{j=1}^{q'}c_j\tau_{x_j}V_4 + V' +C,
\end{equation}
where \(V'\in L^p(\mathbb{R}^d)\) has compact support disjoint from \(\supp V_3\).

Fix \(0<\veps<1\). Let \(\phi_g\) denote the normalized ground state of the operator
\begin{equation}
-\Delta+g\sum_{j=1}^{q'} c_j b_{\sigma(j)}V_3.
\end{equation}
By the variational principle, we have
\begin{equation}
\begin{split}
&(1-\veps)E_{-,\bm{c}_{q'},\bm{x}_{q'}}\left(\frac{g}{1-\veps}\right)-E_{-,3}\Bigg(\sum_{j=1}^{q'} c_j b_{\sigma(j)}g\Bigg) +gC\\
&\geq -\frac{\veps}{q'}\sum_{j=1}^{q'}\tilde{E}_{-,4}\left( \frac{c_j q'}{\veps}g\right)-g\langle V'\phi_g, \phi_g\rangle.
\end{split}
\end{equation}
We now estimate each term on the right-hand side separately as \(g \to +\infty\).

Since \(g_3(E)=o(\tilde{g}_4(E))\) as \(E\to+\infty\) by Condition \hyperref[cd:V6]{\textbf{(V6)}}, for each \(j=1,\ldots, q'\), we obtain
\begin{equation}
\tilde{E}_{-,4}\left( \frac{c_j q'}{\veps}g\right) \leq E_{-,3}\Bigg(\sum_{i=1}^{q'} c_i b_{\sigma(i)}g\Bigg),
\end{equation}
for sufficiently large \(g>0\).

Since \(\supp V_3\) and \(\supp V'\) are compact and disjoint, there exists an open set \(D\subset\mathbb{R}^d\) containing \(\supp V_3\) such that \(D\cap \supp V' =\emptyset\). Lemma \ref{lem:decay} guarantees that the ground state \(\phi_g\) satisfies \eqref{eq:decay} for all \(x\in \mathbb{R}^d\setminus D\). Since \(V'\) has compact support, we find that \(\|V'\|_{L^1(\mathbb{R}^d)}\leq C\|V'\|_{L^p (\mathbb{R}^d)}\). Consequently, for sufficiently large \(g>0\), we have
\begin{equation}
|\langle V'\phi_g, \phi_g\rangle| \leq C\exp\left(-\frac{g^{1/2}}{C}\right) \|V'\|_{L^p(\mathbb{R}^d)}.
\end{equation}

Combining these estimates and \(\lim_{g\to+\infty} E_{-,3}(g)/g =+\infty\) (see \cite{KP}), we find that as \(g\to+\infty\),
\begin{equation}\label{eq:E estimate}
E_{-,\bm{c}_{q'},\bm{x}_{q'}}\left(\frac{g}{1-\veps}\right)
\geq E_{-,3}\Bigg(\sum_{j=1}^{q'} c_j b_{\sigma(j)}g\Bigg) + o(E_{-,3}(g)).
\end{equation}
Since \(\lim_{g\to+\infty}E_{-,3}(g)=+\infty\), this estimate implies that the infimum of the spectrum \(E_{\bm{c}_{q'},\bm{x}_{q'}}(g)\) of the operator \(-\Delta+g\sum_{j=1}^{q'} c_j\tau_{x_j}V\) is negative for sufficiently large \(g>0\). Thus, we have \(\essinf \sum_{j=1}^{q'} c_j\tau_{x_j}V <0\) and \((\bm{c}_{q'},\bm{x}_{q'})\in D_{q'}\). Moreover, \eqref{eq:E estimate} yields
\begin{equation}
\varliminf_{E\to+\infty} \frac{g_{3}(E/(1-\veps))}{g_{\bm{c}_{q'},\bm{x}_{q'}}(E)}\geq (1-\veps)\sum_{j=1}^{q'}c_j b_{\sigma(j)}.
\end{equation}
Since \(g_3\) is concave, from an argument similar to that of the proof of Theorem \ref{thm:pairwise} \textbf{(b)}, we obtain
\begin{equation}
\varliminf_{E\to+\infty}\frac{g_3(E)}{g_3(E/(1-\veps))} \geq 1-\veps.
\end{equation}
Combining these limits, we obtain \eqref{eq:g estimate}.
\end{proof}

\begin{proof}[Proof of Theorem \ref{thm:V5} \textbf{(b)}]
Applying Theorem \ref{thm:V5} \textbf{(a)} to the present setting, Condition \hyperref[cd:V-P1]{\textbf{(V--P1)}} and a straightforward calculation yield the upper bound:
\begin{equation}\label{eq:V6 upper}
\log N(-E)\leq -\frac{a_0}{2\max_{J\in K(I)}\sum_{i\in J}b_i^2}\,g_3(E)^2 (1+o(1))\quad (E\to+\infty).
\end{equation}

We now estimate the lower bound.
Fix \(J\in K(I)\). We can write \(J=\{\sigma(j)\mid j=1,\ldots,q'\}\) for some \(\sigma\in S_q\). We set \((\bm{c}_{q'}, \bm{x}_{q'})=((c_j)_{j=1}^{q'}, (-y_{\sigma(j)})_{j=1}^{q'})\), where each \(c_j\) is given by
\begin{equation}
c_j= \frac{b_{\sigma(j)}}{\sum_{i=1}^{q'}b_{\sigma(i)}}.
\end{equation}

Fix \(\veps>0\). Since \(y_{\sigma(i)}-y_{\sigma(j)}\notin \supp \varphi,\  (i\neq j)\), Theorem \ref{thm:pairwise} \textbf{(a)} and Lemma \ref{lem:V6} imply that for sufficiently large \(E>0\),
\begin{equation}
\begin{split}
\log N(-E) &\geq -(1+\veps)\frac{a_0}{2}\sum_{j=1}^{q'} c_j^2\ g_{\bm{c}_{q'},\bm{x}_{q'}}(E)^2\\
&\geq -(1+\veps)^2 \frac{a_0}{2}\sum_{j=1}^{q'} c_j^2 \Bigg(\sum_{i=1}^{q'} b_{\sigma(i)}c_i\Bigg)^{-2}\ g_{3}(E)^2\\
&\geq -(1+\veps)^2 \frac{a_0}{2} \Bigg( \sum_{j=1}^{q'}b_{\sigma(j)}^2 \Bigg)^{-1} \ g_3(E)^2.
\end{split}
\end{equation}
Therefore, we have the lower bound:
\begin{equation}\label{eq:V6 lower}
\log N(-E) \geq -\frac{a_0}{2\max_{J\in K(I)}\sum_{j\in J} b_j^2}\, g_3(E)^2 (1+o(1))\quad (E\to+\infty).
\end{equation}
This completes the proof.
\end{proof}

\subsection{Proof of Proposition \ref{prop:V6}}

\begin{proof}[Proof of Proposition \ref{prop:V6} \textbf{(a)}]
Fix \(\veps>0\). We choose \(a\in A\) satisfying
\begin{equation}
\frac{a}{\max_{J\in K(\tilde{I}_a)}\sum_{i\in J}b_i^2}\geq \sup_{a'\in A}\frac{a'}{\max_{J\in K(\tilde{I}_{a'})}\sum_{i\in J}b_i^2} -\veps.
\end{equation}
We fix \(r>0\) such that the following hold:
\begin{itemize}
\item \(|y_i-y_j|>2r\) for every \(1\leq i<j \leq q\);
\item \(B(0,2r)\subset \Int S_a\);
\item \(\tilde{I}_a= I_{a,r}\),
\end{itemize}
where we set
\begin{equation}
I_{a,r}=\{(i,j)\in\{1,\ldots,q\}^2 \mid B(y_i-y_j, 2r)\subset \Int  S_a\}.
\end{equation}

For each \(i=1,\ldots,q\), we define \(g_{3,r}^{(i)}(E)\) to be the inverse function of \(-E_{3,r}^{(i)}(g)\), defined for sufficiently large \(g\). Here, \(E_{3,r}^{(i)}(g)\) denotes the infimum of the spectrum of the operator \(-\Delta + gV_{3}^{(i)}\cdot \mathbb{1}_{B(y_i,r)}\).
Since \(\essinf V_{3}^{(i)}=-\infty\) and \(V_{3}^{(i)}-V_{3}^{(i)}\cdot \mathbb{1}_{B(y_i,r)}\) is bounded, a standard variational principle argument shows that \(E_{3,r}^{(i)}(g)\sim E_{3}^{(i)}(g)\) as \(g\to+\infty\).
Since \(g_{3,r}^{(i)}\) and \(g_3^{(i)}\) are concave functions, we obtain 
\begin{equation}
g_{3,r}^{(i)}(E)\sim g_{3}^{(i)}(E)\quad (E\to+\infty).
\end{equation}

Note that \(V\) can be rewritten as \(V=V_1'+\sum_{i=1}^{q} b_i V_{3}^{(i)}\cdot \mathbb{1}_{B(y_i,r)} +V_4 +V_5\) for some bounded function \(V_1'\) that exponentially decays as \(|x|\to+\infty\).
Applying Theorem \ref{thm:V5} \textbf{(a)} and using \(\tilde{I}_a= I_{a,r}\), we obtain
\begin{equation}
\begin{split}
\varlimsup_{E\to+\infty} \frac{\log N(-E)}{g_3(E)^2} &\leq -\frac{a}{2\max_{J\in K(I_{a,r})} \sum_{i\in J}b_i^2}\\
&\leq -\frac{1}{2}\sup_{a'\in A} \frac{a'}{\max_{J\in K(\tilde{I}_{a'})} \sum_{i\in J}b_i^2} +\frac{\veps}{2},
\end{split}
\end{equation}
which implies \eqref{eq:V5 IDS}.
\end{proof}

\begin{proof}[Proof of Proposition \ref{prop:V6} \textbf{(b)}]
By the lower estimate in the proof of Theorem \ref{thm:V5} \textbf{(b)}, we have \eqref{eq:V6 lower} in this case. Combining \eqref{eq:V6 lower}, Proposition \ref{prop:V6} \textbf{(a)}, and Condition \hyperref[cd:V-P2]{\textbf{(V--P2)}}, we obtain \eqref{eq:prop V6 IDS}.
\end{proof}

\begin{remark}
Finally, we note that this remark extends Proposition \ref{prop:V6}~\textbf{(b)} to the case where the shapes of the potential wells are slightly modified.
Throughout this remark, suppose that Conditions \hyperref[cd:V1]{\textbf{(V1)}}, \hyperref[cd:V5]{\textbf{(V5)}}, \hyperref[cd:L]{\textbf{(L)}}, and \hyperref[cd:V-P2]{\textbf{(V--P2)}} are satisfied.

The assumption \(V_3^{(i)}=\tau_{y_i}V_3\) in Condition \hyperref[cd:V6]{\textbf{(V6)}} can be relaxed to allow for local perturbations of the potential wells. Instead, we specifically assume that
\begin{equation}\label{eq:W_i V_3}
V_3^{(i)}= \tau_{y_i}(W_i V_3) \quad (i=1,\ldots,q),
\end{equation}
where each \(W_i\) is a compactly supported continuous function on \(\mathbb{R}^d\) satisfying \(W_i(0)=1\). 
Despite these local perturbations, the conclusion of Proposition \ref{prop:V6} \textbf{(b)} remains valid provided that either of the following two conditions holds.

\noindent \textbf{Case 1 (Nonpositive unique singularity):} \\
The potential \(V_3\) satisfies \(\esssup V_3 \leq 0\) and has a unique singular point at the origin.

\noindent \textbf{Case 2 (Polynomial singularity):} \\
The singularity of \(V_3\) at the origin is of polynomial order. In this case, the assumption \(\esssup V_3 \leq 0\) can be removed; namely, \(V_3\) is given by
\begin{equation}\label{eq:V_3 poly}
V_3(x) = h\bigg(\frac{x}{|x|}\bigg)|x|^{-\nu}
\end{equation}
in a neighborhood of the origin, where \(0<\nu<\min\{2,d/2\}\) and \(h\) is a continuous function on \(\mathbb{S}^{d-1}\) with \(\min h < 0\).

\begin{proof}[Proof sketch]
It suffices to prove an analogue of Lemma \ref{lem:V6} under these settings to estimate the lower bound of \eqref{eq:prop V6 IDS}. 

\noindent \textbf{Case 1:} \\
Since the leading term \(g_3(E)\) as \(E\to+\infty\) is determined solely by the behavior of \(V_3\) in a neighborhood of the origin, we may assume without loss of generality that \(\min_i W_i \geq 1-\veps\) on \(\supp V_3\) for any fixed \(\veps>0\). 
Because \(\esssup V_3 \leq 0\), we have the following operator bound:
\begin{equation}
-\Delta + g\sum_{j=1}^{q'}c_j b_{\sigma(j)}  W_{\sigma(j)} V_3 \leq -\Delta +g(1-\veps)\sum_{j=1}^{q'}c_j b_{\sigma(j)} V_3. 
\end{equation}
By the variational principle, this inequality implies that the minimum eigenvalue of the perturbed operator is bounded from above by \(E_3(g(1-\veps)\sum_{j=1}^{q'}c_j b_{\sigma(j)})\),
which is sufficient to establish the analogue of Lemma \ref{lem:V6}.

\noindent \textbf{Case 2:} \\
We set
\begin{equation}
W(x)= \frac{\sum_{j=1}^{q'}c_j b_{\sigma(j)}W_{\sigma(j)}(x)}{\sum_{j=1}^{q'}c_j b_{\sigma(j)}}.
\end{equation}
Let \(E_{3}^{*}(g)\) denote the infimum of the spectrum of the operator \(-\Delta + gWV_3\). If \(V_3\) is given by \eqref{eq:V_3 poly}, we have \(E_{3}^{*}(g) \sim E_{3}(g)\) as \(g\to+\infty\) (see \cite[Lemma 5.7]{KP}). We note that it follows that
\begin{equation}
-\Delta +g \sum_{j=1}^{q'}c_j b_{\sigma(j)}W_{\sigma(j)} V_3 =-\Delta + g \sum_{j=1}^{q'}c_j b_{\sigma(j)} W V_3.
\end{equation}
Combining these facts with simple calculations, we establish the corresponding statement of Lemma \ref{lem:V6}.
\end{proof}
\end{remark}



\begin{thebibliography}{99}

\bibitem{CL}
R. Carmona and J. Lacroix,
\emph{Spectral Theory of Random Schr\"{o}dinger Operators}, Probability and its Applications,
Birkh\"{a}user Boston, Inc., Boston, MA, 1990.
\MR{1102675}

\bibitem{Dereudre}
D. Dereudre,
Introduction to the theory of Gibbs point processes, in \emph{Stochastic Geometry}, 181--229, Lecture Notes in Mathematics, 2237, Springer, Cham, 2019.
\MR{3931586}

\bibitem{DV}
D. Dereudre and T. Vasseur, 
Existence of Gibbs point processes with stable infinite range interaction,
\emph{J. Appl. Probab.} \textbf{57} (2020), 775--791.
\MR{4148057}

\bibitem{Diestel}
R. Diestel,
\emph{Graph Theory} (5th ed.), Graduate Texts in Mathematics, 173,
Springer, Berlin, 2017.
\MR{3644391}

\bibitem{FU}
R. Fukushima and N. Ueki,
Classical and quantum behavior of the integrated density of states for a randomly perturbed lattice,
\emph{Ann. Henri Poincar\'{e}} \textbf{11} (2010), 1053--1083.
\MR{2737491}

\bibitem{Georgii}
H.-O. Georgii,
\emph{Gibbs Measures and Phase Transitions} (2nd ed.), de Gruyter Studies in Mathematics, 9, Walter de Gruyter, Berlin, 2011.
\MR{2807681}

\bibitem{GK}
H.-O. Georgii and T. K\"{u}neth,
Stochastic comparison of point random fields,
\emph{J. Appl. Probab.} \textbf{34} (1997), 868--881.
\MR{1484021}

\bibitem{HH}
C. Hofer-Temmel and P. Houdebert,
Disagreement percolation for Gibbs ball models, \emph{Stochastic Process. Appl.} \textbf{129} (2019), 3922--3940.
\MR{3997666}

\bibitem{KMN}
M. Kaminaga, T. Mine, and F. Nakano,
Integrated density of states for the Poisson point interactions on \(\mathbb{R}^3\),
\emph{Math. Z.} \textbf{311} (2025), 45, 47 pp.
\MR{4951583}

\bibitem{Kirsch}
W. Kirsch,
Random Schr\"{o}dinger operators: a course, in \emph{Schr\"{o}dinger Operators}, 
264--370, Lecture Notes in Physics, 345, Springer, Berlin, 1989.
\MR{1037323}

\bibitem{KP}
F. Klopp and L. Pastur,
Lifshitz tails for random Schr\"{o}dinger operators with negative singular Poisson potential,
\emph{Comm. Math. Phys.} \textbf{206} (1999), 57--103. \MR{1736990}

\bibitem{LO}
G. Last and M. Otto,
Disagreement coupling of Gibbs processes with an application to Poisson approximation,
\emph{Ann. Appl. Probab.} \textbf{33} (2023), 4091--4126. 
\MR{4663505}

\bibitem{LGP}
I. M. Lifshitz, S. A. Gredeskul, and L. A. Pastur,
\emph{Introduction to the Theory of Disordered Systems},
translated from the Russian, Wiley, New York, 1988.
\MR{1042095}

\bibitem{MR}
R. Meester and R. Roy,
\emph{Continuum Percolation}, Cambridge Tracts in Mathematics, 119, Cambridge University Press, Cambridge, 1996.
\MR{1409145}

\bibitem{Nakagawa}
Y. Nakagawa,
Asymptotic behaviors of the integrated density of states for random Schr\"{o}dinger operators associated with Gibbs point processes,
\emph{Electron. J. Probab.} \textbf{28} (2023), 158, 14 pp.
\MR{4671142}

\bibitem{Pastur}
L. A. Pastur,
The behavior of certain Wiener integrals as \(t\to\infty\) and the density of states of Schr\"{o}dinger equations with random potential, 
\emph{Teoret. Mat. Fiz.} \textbf{32} (1977), 88--95 (in Russian). \MR{0449356}

\bibitem{PF}
L. Pastur and A. Figotin,
\emph{Spectra of Random and Almost-Periodic Operators}, Grundlehren der mathematischen Wissenschaften, 297,
Springer, Berlin, 1992.
\MR{1223779}

\bibitem{Preston}
C. Preston,
\emph{Random Fields}, Lecture Notes in Mathematics, 534,
Springer, Berlin-New York, 1976.
\MR{0448630}

\bibitem{RS2}
M. Reed and B. Simon,
\emph{Methods of Modern Mathematical Physics. II. Fourier Analysis, Self-Adjointness},
Academic Press, New York-London, 1975.
\MR{0493420}

\bibitem{RS4}
M. Reed and B. Simon,
\emph{Methods of Modern Mathematical Physics. IV. Analysis of Operators}, 
Academic Press, New York-London, 1978.
\MR{0493421}

\bibitem{Ruelle}
D. Ruelle, 
\emph{Statistical Mechanics: Rigorous Results}, Reprint of the 1989 edition,
World Scientific Publishing Co., Inc., River Edge, NJ; Imperial College Press, London, 1999.
\MR{1747792}

\bibitem{Strauss}
D. J. Strauss,
A model for clustering,
\emph{Biometrika} \textbf{62} (1975), 467--475. \MR{0383493}

\end{thebibliography}


\begin{acks}
The author would like to thank Professor Naomasa Ueki for the helpful suggestions. 
The author is grateful to Kumano Dormitory, Kyoto University for its generous financial and living assistance.
Finally, the author acknowledges the use of Gemini (Google) for its assistance in refining the English phrasing of this paper.
This work was supported by JST SPRING, Grant Number JPMJSP2110.
\end{acks}


\end{document}